\documentclass[a4paper,11pt]{article}

\pdfoutput=1
\usepackage{jheppub} %
\usepackage{lineno}
\usepackage{verbatim}
\usepackage{amsmath}
\usepackage{amsthm}
\usepackage{braket}
\usepackage{dsfont}
\usepackage{accents}
\usepackage[normalem]{ulem}
\usepackage[inkscapelatex=false]{svg}
\newtheorem{theorem}{Claim}[section]
\newtheorem{definition}{Definition}

\definecolor{acolor}{rgb}{0.1,.6,0}

\definecolor{rcolor}{rgb}{0.9,0.1,0.1}

\title{When left and right disagree: Entropy and von Neumann algebras in quantum gravity with general AlAdS boundary conditions}

\author[a]{Donald Marolf} \emailAdd{marolf@ucsb.edu}

\author[b]{and Daiming Zhang}\emailAdd{zhang-dm20@mails.tsinghua.edu.cn}

\affiliation[a]{Department of Physics, University of California, Santa Barbara, CA 93106, USA}
\affiliation[b]{Tsinghua University, 100084
Beijing, China}

\abstract{Euclidean path integrals for UV-completions of $d$-dimensional bulk quantum gravity were recently studied in \cite{Colafranceschi:2023urj} by assuming that they satisfy axioms of finiteness, reality, continuity, reflection-positivity, and factorization.  Sectors ${\cal H}_{\cal B}$ of the resulting Hilbert space were then defined for any $(d-2)$-dimensional surface ${\cal B}$, where ${\cal B}$ may be thought of as the boundary $\partial \Sigma$ of a bulk Cauchy surface in a corresponding Lorentzian description, and where ${\cal B}$ includes the specification of appropriate boundary conditions for bulk fields.  Cases where ${\cal B}$ was the disjoint union $B \sqcup B$ of
two identical  $(d-2)$-dimensional surfaces $B$ were studied in detail  and, after the inclusion of finite-dimensional `hidden sectors,' were shown to provide a Hilbert space interpretation of the associated Ryu-Takayanagi entropy. The analysis was performed by constructing type-I von Neumann algebras $\mathcal A_L^B, \mathcal A_R^B$ that act respectively at the left and right copy of $B$ in $B\sqcup B$. %

Below, we consider the case of general ${\cal B}$, and in particular for ${\cal B} = B_L \sqcup B_R$ with $B_L, B_R$ distinct.
For any $B_R$, we find that the von Neumann algebra at $B_L$ acting on the off-diagonal Hilbert space sector ${\cal H}_{B_L \sqcup B_R}$ is a central projection of the corresponding type-I 
 von Neumann algebra on the `diagonal' Hilbert space  ${\cal H}_{B_L \sqcup B_L}$. As a result, the von Neumann algebras $\mathcal A_L^{B_L},\mathcal A_R^{B_L}$ defined in \cite{Colafranceschi:2023urj} using the diagonal Hilbert space $\mathcal H_{B_L \sqcup B_L}$ turn out to coincide precisely with the analogous algebras defined using the full Hilbert space of the theory (including all sectors $\mathcal H_{\mathcal B}$).  A second implication is that, for any ${\cal H}_{B_L \sqcup B_R}$,  including  the same hidden sectors as in the diagonal case again provides a Hilbert space interpretation of the Ryu-Takayanagi entropy. We also show the above central projections to satisfy consistency conditions that lead to a universal central algebra relevant to all choices of $B_L$ and $B_R$.
}

\begin{document}
\maketitle
\flushbottom

\section{Introduction}
\label{sec:intro}

As emphasized in \cite{Colafranceschi:2023urj}, a number of arguments regarding gravitational entropy that were originally motivated by the AdS/CFT correspondence have now been understood to follow directly from bulk physics.   A primary example is 
the derivation \cite{Almheiri:2019qdq,Penington:2019kki} of the Island Formula for the entropy of  Hawking radiation transferred from an asymptotically-locally-AdS (AlAdS) gravitational system to a non-gravitational bath.  This derivation simply combines the gravitational path integral arguments of \cite{Lewkowycz:2013nqa,Faulkner:2013ana,Dong:2016hjy,Dong:2017xht} with the setting studied in \cite{Penington:2019npb,Almheiri:2019psf}. And as explained in \cite{Marolf:2020rpm}, in this context the results may be safely interpreted in terms of standard von Neumann entropies without invoking holography at any intermediate step.  

Another class of examples  involves taking the semiclassical limit in which Hilbert space densities of states diverge.  Using purely bulk methods, one can show that the algebra of observables is generated by a type-II von Neumann factor and its commutant.  
This observation then leads to an entropy on these algebras that agrees with the quantum-corrected RT formula up to an additive constant \cite{Witten:2021unn,Chandrasekaran:2022eqq,Penington:2023dql,Kudler-Flam:2023qfl}.

Motivated by such results, it was suggested in \cite{Colafranceschi:2023urj} that purely bulk arguments (i.e., without assuming the existence of a holographic dual field theory) should suffice to provide a Hilbert-space interpretation of an entropy 
defined by regions of an AlAdS boundary for which the semiclassical limit is given by
the Ryu-Takayangi formula \cite{Ryu:2006bv,Ryu:2006ef} (or its covariant Hubeny-Rangamani-Takayanagi generalization \cite{Hubeny:2007xt}).

By assuming certain Euclidean-signature axioms, this was then shown to be the case in 
so-called `diagonal' settings where the boundary $\partial \Sigma$ of a Cauchy surface  $\Sigma$ in a corresponding Lorentz-signature spacetime took the form $\partial \Sigma = B \sqcup B$, where $B$ was a compact $(d-2)$-dimensional manifold (with $\partial B=\emptyset$), and on which appropriate boundary conditions were specified for bulk fields.

The argument of \cite{Colafranceschi:2023urj} was formulated in terms of a supposed path integral for a UV-complete finite-coupling\footnote{Gravitational path integrals associated with familiar classical actions diverge in the limit $G\rightarrow 0$ (or $S_0\rightarrow \infty$ for JT gravity).  As a result, the asymptotic expansion of such path integrals in powers of $G$ will generally fail to satisfy the axioms of \cite{Colafranceschi:2023urj}; see \cite{Colafranceschi:2023txs} and especially section 5.4 of \cite{Colafranceschi:2023urj} for further discussion of this issue.} bulk asymptotically-locally-AdS$_d$ (AlAdS$_d$) theory.  The previously-advertised axioms for such path integrals were called finiteness, reality\footnote{The reality axiom implies that the theory to be invariant under a notion of time-reversal symmetry.  We expect that this axiom is not in fact necessary.}, reflection-positivity, continuity, and factorization. In the above diagonal setting, these properties sufficed to show the associated Hilbert space sectors $\mathcal H_{B \sqcup B}$ to be direct sums 
$\bigoplus_\mu \mathcal H_{B \sqcup B}^\mu$ of Hilbert spaces that factorize as $\mathcal H_{B \sqcup B}^\mu =\mathcal H_{B \sqcup B,L}^\mu \otimes \mathcal H_{B \sqcup B,R}^\mu$, where the left and right factors $\mathcal H_{B \sqcup B,L}^\mu$ and $\mathcal H_{B \sqcup B,R}^\mu$ for given $\mu$ are isomorphic but are associated with operator algebras that, in an appropriate sense, act at the respective left and right copy of $B$ in $\mathcal H_{B \sqcup B}$.  It was then further shown that the path integral defines a trace on such algebras that agrees with the standard sum-over-diagonal-matrix-elements Hilbert-space trace associated with the Hilbert space
\begin{eqnarray}
    \label{eq:tildeHs}
\undertilde {\cal H}_{B} = \bigoplus_\mu {\cal H}_{B \sqcup B, L}^\mu \otimes {\mathbb C}^{n_\mu} = \bigoplus_\mu {\cal H}_{B \sqcup B, R}^\mu \otimes {\mathbb C}^{n_\mu}
\end{eqnarray}
for appropriate integers $n_\mu.$  The corresponding entropies thus agree as well and, by first making use of an appropriate embedding of ${\cal H}_{B \sqcup B}$ in $\undertilde{\mathcal H}_B \otimes \undertilde{\mathcal H}_B$ and then connecting with the Lewkowycz-Maldacena argument \cite{Lewkowycz:2013nqa} and its generalizations, one obtains a Hilbert space interpretation of the Ryu-Takayanagi entropy of either boundary $B$.

The reader should note that the Hilbert space $\undertilde {\cal H}_{B}$ was not explicitly  introduced in \cite{Colafranceschi:2023urj}, though its use will simplify our discussion.  
Indeed, the tilde decoration at the bottom of $\undertilde {\cal H}_{B}$ is intended to  help to distinguish $\undertilde {\cal H}_{B}$ from the various Hilbert spaces defined in \cite{Colafranceschi:2023urj} that were denoted by symbols with upper tildes.
As shown by the 2nd equality in \eqref{eq:tildeHs}, it would be unnatural to assign either an $L$ or an $R$ label to $\undertilde {\cal H}_{B}$.  Furthermore, in the diagonal case we can lose nothing by using only a subscript $B$ instead of $B\sqcup B$.  This will turn out to be a good choice of notation as we will see below that the same space $\undertilde {\cal H}_{B}$ arises in the analogous analysis of any off-diagonal sector ${\cal H}_{B \sqcup B'}$, where now 
${\cal H}_{B \sqcup B'}$ is to be embedded in $\undertilde {\cal H}_{B} \otimes  \undertilde {\cal H}_{B'}$.

The focus of \cite{Colafranceschi:2023urj} on diagonal sectors $\mathcal H_{B\sqcup B}$ had two primary drawbacks.  The most obvious was of course that it provided the desired Hilbert space interpretation of RT entropy only when the two boundaries are identical.  Directly generalizing the arguments of \cite{Colafranceschi:2023urj} to the case with $(d-2)$-dimensional boundary $B \sqcup B'$ turns out to be nontrivial due to the reliance of \cite{Colafranceschi:2023urj} on special properties of cylinders $C_\epsilon = B \times [0,\epsilon]$.  The obstacle is that such cylinders are  intrinsically diagonal in the sense that $\partial C_\epsilon = B \sqcup B$. We will thus seek other arguments below.

The second issue arises from the fact that the choice of Hilbert space plays a role in the construction of the desired operator algebras.  In particular, while an algebra of simple operators can be defined directly using smooth boundary conditions in the path integral, the most useful mathematical structures turn out to be the left and right von Neumann algebras constructed by using a Hilbert space to define an appropriate notion of a completion.  While this may seem like a mathematical technicality, it raises the interesting question of whether using the full quantum gravity Hilbert space might lead to von Neumann algebras (and thus to entropies) that differ from the ones constructed using only the diagonal sectors $\mathcal H_{B \sqcup B}$ as in \cite{Colafranceschi:2023urj}.

We will address both of these shortfalls below.  As further motivation for our study,  it is useful to take inspiration from the AdS/CFT correspondence, in which the bulk path integral is simply equal to a path integral for the dual CFT.  While it is not necessarily the most general allowed setting, this case certainly satisfies the axioms of \cite{Colafranceschi:2023urj}.  Furthermore, in the AdS/CFT context, for any choice of $B$ and $B'$ we always have the so-called Harlow-factorization property\footnote{The name comes from the emphasis on this property in \cite{Harlow:2015lma}.} $\mathcal H_{B \sqcup B'} = \mathcal H_{B}\otimes \mathcal H_{B'}$ so that in particular, the left Hilbert space factor (and the associated type-I von Neumann factor) is manifestly independent of the choice of $B'$. \footnote{Indeed, since for AdS/CFT there is only a single value of $\mu$ and it has $n_\mu=1$, for this case \eqref{eq:tildeHs} gives  $\undertilde{\mathcal H}_B=\mathcal H_B.$}

One might then hope that a similar independence of $B'$ follows more generally from the axioms of \cite{Colafranceschi:2023urj}.  Below, it will be useful to rename $B,B'$ as $B_L, B_R$ and to refer to $B_L, B_R$ as the left and right boundaries.  In that notation, while we have already seen in the diagonal context that a given boundary $B_L$ is associated with a {\it set} of left Hilbert space factors $\mathcal H_{B_L \sqcup B_L, L}$, one may nevertheless hope that the full set of such factors has no dependence on $B_R$.

As foreshadowed above, 
this will turn out to be nearly true in the sense that the off-diagonal Hilbert spaces $\mathcal H_{B_L \sqcup B_R}$ again decompose according to\footnote{Here we use notation chosen to mirror that of \cite{Colafranceschi:2023urj}.  In this notation, we emphasize that an object labelled with $B_L\sqcup B_R$ may in fact depend on the partition of ${\cal B} = B_L\sqcup B_R$ into $B_L$ and $B_R$, and thus that it is generally not determined entirely by ${\cal B}$ alone.  Symbols in which an explicit $\sqcup$ does not appear will be free of this issue.}
\begin{eqnarray}
\label{eq:LRsumintro}
    \mathcal H_{B_L \sqcup B_R} &=& \bigoplus_\mu \mathcal H_{B_L \sqcup B_R}^\mu \ \ \ {\rm with} \ \ \
    \mathcal H_{B_L \sqcup B_R}^\mu = \mathcal H_{B_L \sqcup B_R,L}^\mu \otimes \mathcal H_{B_L \sqcup B_R,R}^\mu,
\end{eqnarray}
such that every left Hilbert space factor $\mathcal H_{B_L \sqcup B_R,L}^\mu$ is in fact canonically isomorphic to a left factor $\mathcal H_{B_L \sqcup B_L, L}^\mu$ associated with the diagonal Hilbert space 
$\mathcal H_{B_L \sqcup B_L}$. However, for given $B_R$ it may be that we find only a subset of the diagonal left Hilbert-space factors $\mathcal H_{B_L \sqcup B_L, L}^\mu$. Indeed, we will see that there is a natural sense in which the isomorphic factors $\mathcal H_{B_L \sqcup B_R,L}^\mu$ and $\mathcal H_{B_L \sqcup B_L,L}^\mu$ can be said to be associated with the same value of $\mu$, so that for any $\mu$ appearing in \eqref{eq:LRsumintro} it is natural to write
\begin{equation}
\label{eq:HmuBL}
\mathcal H_{B_L \sqcup B_R,L}^\mu = \mathcal H_{B_L \sqcup B_L,L}^\mu =: \mathcal H_{B_L}^\mu, 
\end{equation}
where this defines the symbol $\mathcal H_{B_L}^\mu$ used on the right-hand-side. Here we have refrained from adding an additional ${}_{,L}$ subscript on $\mathcal H_{B_L}^\mu$ since the identical Hilbert space arises from the corresponding construction when $B_L$ is used as a right boundary (instead of a left boundary as above).  In addition, for given $\mu$ the integer $n_\mu$  will be shown to be independent of the choice of boundary so that, in the notation of \eqref{eq:HmuBL}, replacing $B$ in \eqref{eq:tildeHs} by $B_L, B_R$ we may write both
\begin{equation}
\label{eq:utLR}
\undertilde {\mathcal H}_{B_L} := \bigoplus_\mu \left( \mathcal H_{B_L}^\mu \otimes {\mathbb C}^{n_\mu} \right) \ \ \ {\rm and}
\ \ \ \undertilde {\mathcal H}_{B_R} := \bigoplus_\mu  \left( {\mathcal H}_{B_R}^\mu \otimes {\mathbb C}^{n_\mu}\right) 
\end{equation}
for the same integers $n_\mu$.

The results of \cite{Colafranceschi:2023urj} then imply that the trace defined by the path integral on operators that act at any $B_L$ or $B_R$ coincides with the sum-over-diagonal-matrix-elements trace defined by the Hilbert spaces  \eqref{eq:utLR}.  A Hilbert space interpretation of the Ryu-Takayanagi entropy associated with either $B_L$ or $B_R$ of states in ${\cal H}_{B_L \sqcup B_R}$ then follows from an appropriate embedding of ${\cal H}_{B_L \sqcup B_R}$ in $\undertilde {\cal H}_{B_L} \otimes \undertilde {\cal H}_{B_R}$.

Finally, we will also verify the analogous statements for the above-mentioned von Neumann algebras, showing in particular that the algebras constructed in \cite{Colafranceschi:2023urj} using only the diagonal sectors $\mathcal H_{B \sqcup B}$ do in fact coincide with von Neumann algebras completed by using the topology defined by the entire quantum gravity Hilbert space.
More specficially, we will see that the von Neumann algebra acting at $B_L$ associated with an off-diagonal sector $\mathcal H_{B_L \sqcup B_R}$ is always a central projection of the corresponding 
algebra defined by the diagonal Hilbert space sector $\mathcal H_{B_L \sqcup B_L}$.  Furthermore, these projections will be shown to satisfy compatibility conditions that allow us to assemble such central projections into a universal central algebra, independent of the choice of any particular $B_L$, from which the central algebra for each pair $B_L$,$B_R$ can be recovered by acting with an appropriate projection. Such algebraic results are in fact more fundamental than the Hilbert-space results described above and will thus be addressed first in the work below.

This paper is organized as follows. We begin in section \ref{sec:rev} by reviewing the construction of algebras and Hilbert spaces from Euclidean path integrals as described in~\cite{Colafranceschi:2023urj}. This includes the definition of general off-diagonal Hilbert space sectors $\mathcal H_{B_L \sqcup B_R}$, as well as algebras  $\hat A_L^{B_L \sqcup B_R}$, $\hat A_R^{B_L \sqcup B_R}$ of operators on $\mathcal H_{B_L \sqcup B_R}$ defined by attaching surfaces respectively to the left and right boundaries $B_L,B_R$.  However, these algebras are not complete in any natural topology, and \cite{Colafranceschi:2023urj} defined von Neumann completions only in the diagonal context $B_L=B_R$.  The new results begin in section \ref{sec:alg}, which shows that the off-diagonal left-algebra $\hat A_L^{B_L \sqcup B_R}$ is canonically identified with a quotient of the diagonal left-algebra 
$\hat A_L^{B_L \sqcup B_R}$, and similarly for the right-algebras.  It then remains to study the completions that define the off-diagonal von Neumann algebras in section \ref{sec:hom}.  After developing some useful technology, we demonstrate that the off-diagonal von Neumann algebras are again canonically identified with quotients of the diagonal von Neumann algebras.  Section \ref{sec:proj} then shows this identification to take the form of a central projection, discusses the relationship between the off-diagonal and diagonal Hilbert spaces, and organizes the discussion of centers in terms of a universal central algebra that is independent of the choices of boundaries. The fact that the left and right von Neumann algebras are commutants on $\mathcal H_{B_L \sqcup B_R}$ is also established in this section by making use of further supporting results from appendix \ref{app:Riefel}.  With all of the above results in place, it is then straightforward to describe the Hilbert space interpretation of RT entropy in the off-diagonal context.  This is done in  section \ref{sec:TrEnt}, after which further discussion and final comments are provided in section \ref{sec:disc}.

\section{Algebras and Hilbert spaces from gravitational path integrals}
\label{sec:rev}

The results of \cite{Colafranceschi:2023urj} were established within an axiomatic framework for the Euclidean path integral in UV completions of quantum theories of gravity.  The five axioms used in \cite{Colafranceschi:2023urj} are briefly summarized below, though we refer the reader to \cite{Colafranceschi:2023urj} for full details and additional discussion.
\begin{enumerate}
    \item \textbf{Finiteness:}
    The boundary conditions for the path integral are assumed to form a space $X^d$ of \(d\)-dimensional `source manifolds' \(X^d\).  The path integral then defines a map \(\zeta : X^d \to \mathbb{C}\) to the complex numbers; i.e., \(\zeta(M)\) is well-defined and finite for every \(M \in X^d\).  Local restrictions on the sources may be imposed as needed to achieve this property.  As an example, one may choose to require source manifolds to have non-negative scalar curvature.
    \item \textbf{Reality:} Let $\underline X^d$ denote formal finite linear combinations of source manifolds with coefficients in $\mathbb C$. We extend $\zeta$ to elements of $\underline X^d$ by linearity.  For every \( M \in \underline X^d \), we have both \( M^* \in \underline X^d \) and \( [\zeta(M)]^* = \zeta(M^*) \).  This axiom is trivial if the original space $X^d$ of source manifolds is taken to be real; i.e., if ${}^*$ is taken to act trivially on $X^d$.  Furthermore, as noted in the introduction, this axiom also implies a time-reversal symmetry.  We thus expect that it can be dropped without significant harm, though we leave such a study for future work.
    \item \textbf{Reflection-Positivity:} $\zeta(M)$ is a non-negative real number for reflection-symmetric source manifolds $M$, i.e. \( M \in \underline X^d \) can be written in the form \( M = \sum_{I,J=1}^{n} \gamma_I^* \gamma_J M_{I,J} \) where \( \gamma_I \in \mathbb{C} \), \( \gamma_I^* \) denotes the complex conjugate of \( \gamma_I \), and where each \( M_{I,J} \) can be sliced into two parts \( N_I^* \) and \( N_J \), for some \( n \in \mathbb{Z}^+ \).
    \item \textbf{Continuity:} Suppose that the source manifold $M\in X^d$ contains a  ‘cylinder’ \( C_\epsilon \) of the form \( B \times [0, \epsilon] \).  Then \(\zeta\) is continuous in the length \( \epsilon \) of this cylinder for all $\epsilon >0$.
    \item \textbf{Factorization:} For closed boundary manifolds \( M_1, M_2 \) and their disjoint union \( M_1 \sqcup M_2 \), we have \(\zeta(M_1 \sqcup M_2) = \zeta(M_1)\zeta(M_2) \).
\end{enumerate}
The framework can also be applied to contexts like those in \cite{Saad:2019lba} and \cite{Marolf:2020xie} where factorization fails, but where the path integral can be expressed as an integral ($\zeta = \int d\alpha \ \zeta_\alpha$) over path integrals $\zeta_\alpha$ in which all of the above axioms hold.  The results of \cite{Colafranceschi:2023urj} then clearly hold for each $\zeta_\alpha$, with corresponding implications for the full path integral $\zeta$.

Another important ingredient in the discussion of
\cite{Colafranceschi:2023urj} was the concept of a source-manifold $N$ with boundary $\partial N$. An operation ${}^*$ (also used in axiom 3) was defined on $Y_{\mathcal B}^d$ by complex-conjugating the sources on $N$ and simultaneously reflecting $N$ about its boundary.  This ${}^*$ was then used to define 
the notion of a {\it rimmed} source-manifold-with-boundary $N$, which is a was allowed to have a non-trivial boundary $\partial N$ so long as some open set containing $\partial N$ was a cylinder of the form $C_\epsilon = B \times [0, \epsilon]$  described above with $C_\epsilon = C_\epsilon^*$.  We will discuss only boundaries $B$ for which there exist cylinders satisfying this condition.

The space of such rimmed source-manifolds with boundary $\mathcal B$ was denoted $Y^d_{\mathcal B}$. We see that any $N_1,N_2\in Y^d_{\mathcal B}$ can be naturally sewn together across $\mathcal B$ to define a smooth closed source-manifold $M_{N_1^*N_2}$. The space of formal finite linear combinations $\underline Y_{\mathcal B}^d$,  equipped with a pre-inner product defined by the path integral of glued source manifolds $\braket{N_1|N_2}=\zeta(M_{N_1^*N_2})$, then forms a pre-inner product space\footnote{This is the same $H_{\mathcal B}$ as in \cite{Colafranceschi:2023urj}, where it was called a pre-Hilbert space.} which we denote as $H_{\mathcal B}$. The Hilbert space $\mathcal H_{\mathcal B}$ is obtained by first taking the quotient by the space of null vectors $\mathcal N_{\mathcal B}$, and then taking the completion with respect to the norm. To reflect the fact that it is a dense subspace of $\mathcal H_{\mathcal B}$,
we introduce the notation $\mathcal D_{\mathcal B}=H_{\mathcal B}/\mathcal N_{\mathcal B}$ for the pre-Hilbert space defined before taking the completion.

The analysis of~\cite{Colafranceschi:2023urj} focused on the case when the boundary ${\mathcal B}$ is a disjoint union of two closed boundary manifolds ${\mathcal B}=B_1\sqcup B_2$.  The spaces $\underline Y^d_{B_1\sqcup B_1}$ and $\underline Y^d_{B_2\sqcup B_2}$ can then be promoted to algebras $A_L^{B_1}, A_R^{B_2}$ by defining products that simply glue together the two surfaces being multiplied.  The product $a\cdot_L b$ on $\underline Y^d_{B_1\sqcup B_1}$ (which is used to define $A_L^{B_1}$) is defined by gluing the right boundary of $a$ to the left boundary of $b$, while
the product $c\cdot_R d$ on $\underline Y^d_{B_2\sqcup B_2}$ (which is used to define $A_R^{B_2}$) is defined by gluing the left boundary of $c$ to the right boundary of $d$. Since left and right products are related by $a\cdot_L b=b\cdot_R a$, we may define $ab:=a\cdot_L b=b\cdot_R a$. There is also a natural involution $\star$ defined by $a^\star :=\left(a^{t}\right)^*$,  where the transpose operation ${}^t$ simply interchanges the labels left and right on the boundaries of $a$.  Thus $a^t$ is the same source manifold as $a$ but with the left boundary of $a$ now called the right boundary of $a^t$, and vice versa.  

If we interchange the two boundaries to instead use $\mathcal B = B_2 \sqcup B_1$, the same construction defines analogous algebras $A_L^{B_2}$ and $A_R^{B_1}$.  The involution $\star$ then defines an anti-linear isomorphism between $A^{B_i}_L$ and $A^{B_i}_R$. Furthermore, a trace operation $tr$ can be defined on these surface algebras using the path integral, $tr(a):=\zeta (M(a))$, where $M(a)\in \underline X^d$ denotes the source manifold obtained from gluing together the two copies of $B$ in the boundary of $a\in \underline Y^d_{B\sqcup B}$.

Representations of the surface algebras $A^{B_1}_L$ and $A^{B_2}_R$ on the sector $\mathcal H_{B_1\sqcup B_2}$ were then constructed in two steps. The first step was to consider the natural actions of $A_L^{B_1}, A_R^{B_2}$ on the pre-inner product space $H_{B_1\sqcup B_2}$ by gluing the relevant surfaces along corresponding boundary components ($B_1$ or $B_2$).  For example, $a\in A^{B_1}_L$ is represented by an operator $\hat a_L$ that acts on $\ket b\in H_{B_1\sqcup B_2}$ by gluing the right boundary of $a$ to the left boundary of $b$ so that $\hat a_L \ket b=\ket {ab}$. The next step used the trace inequality
\begin{equation}
\label{eq:trIn}
tr(b^\star a a^\star b) \le tr(a^\star a) tr(b^\star b).
\end{equation}
derived in~\cite{Colafranceschi:2023urj} for\footnote{\label{foot:genTrIn}In fact, the inequality \eqref{eq:trIn} was derived in \cite{Colafranceschi:2023txs} for any $a\in \underline Y^d_{B_1 \sqcup B_2}$, $b\in \underline Y^d_{B_1 \sqcup B_3}$.  In that context, we can still define a corresponding operation $\star$ such that 
$a^\star\in \underline Y^d_{B_2 \sqcup B_1}$, $b^\star\in \underline Y^d_{B_3 \sqcup B_1}$, and concatenation of surfaces then defines $a^\star a  \in A_L^{B_2}$ and $b^\star b, b^\star a a^\star b \in A_L^{B_3}$.  This more general version will be useful in appendix \ref{app:Psi}. } $a$ in either $A^{B_1}_L$ or $A^{B_2}_R$ and any $b \in H_{B_1\sqcup B_2}$.  As shown in figure~\ref{fig:ineq}, the relation \eqref{eq:trIn} is equivalent to the inequality
 \begin{equation}
\label{eq:bounded}
\langle b| \hat a_L \hat a_L^\dagger | b \rangle \le tr(a^\star a) \langle b|b\rangle,
\end{equation}
which immediately implies that each operator in the representation is bounded.  It thus preserves the null space $\mathcal N_{B_1\sqcup B_2}$, and induces a (bounded) operator on $\mathcal D_{B_1\sqcup B_2}$.  It follows that there is a unique bounded extension to the Hilbert space $\mathcal H_{B_1\sqcup B_2}$.

\begin{figure}[htbp]
\centering
\begin{minipage}[b]{0.8\textwidth}
  \centering
  \includegraphics[clip, trim=0.5cm 20cm 0.5cm 11cm, width=\linewidth]{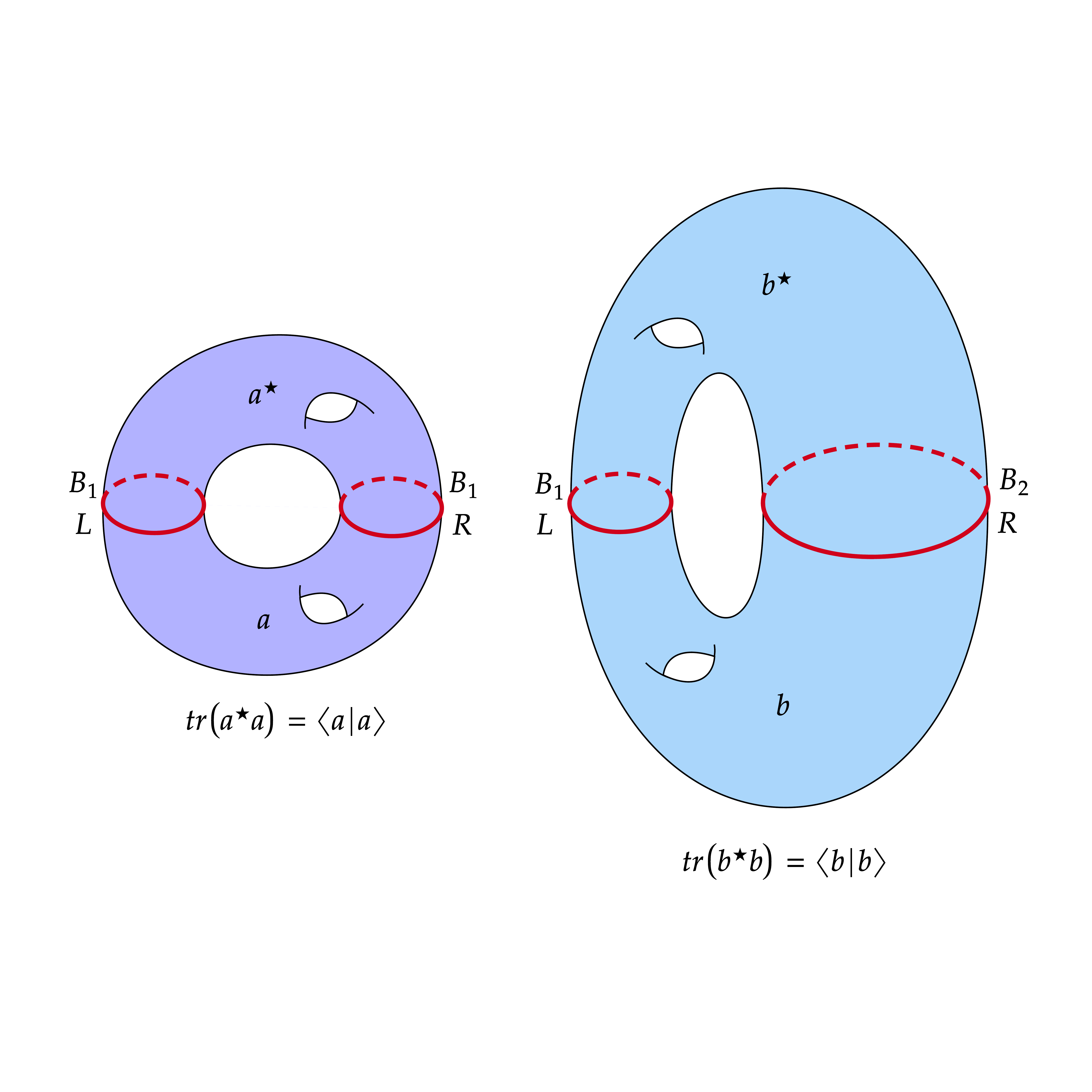}
\end{minipage}

\vspace{-1cm}

\begin{minipage}[b]{0.6\textwidth}
  \centering
  \includegraphics[width=\linewidth]{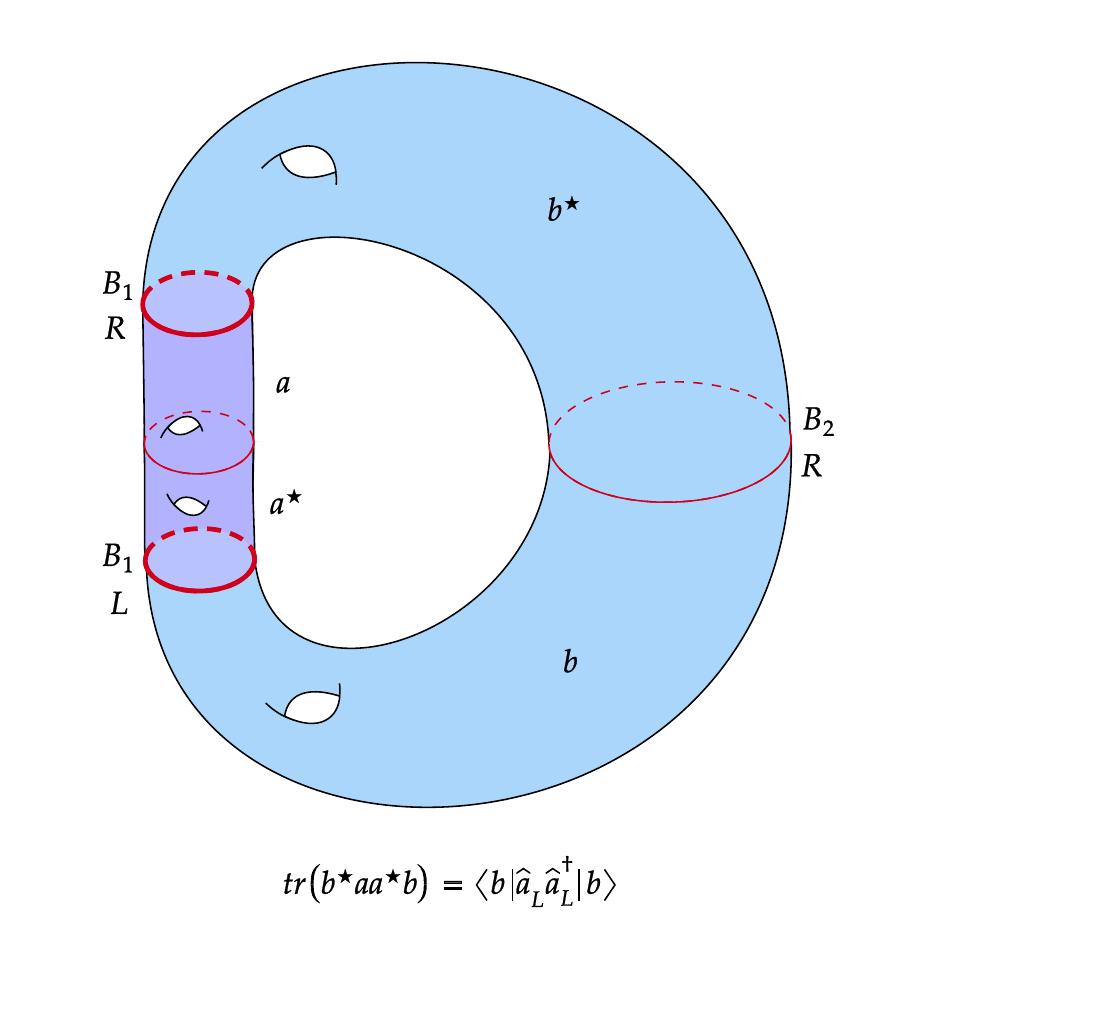}
\end{minipage}
\caption{For surfaces $a \in Y^d_{B_L\sqcup B_L}$ and $b \in Y^d_{B_L\sqcup B_R}$, the traces of   $a^\star a$ and  $b^\star b$  coincide with  $\langle a|a\rangle$ and $\langle b|b\rangle$ as shown in the upper panel. The left hand side $tr(b^\star a a^\star b)$ of the trace inequality \eqref{eq:trIn} computes the inner product $\langle b|\hat a_L\hat a_L^\dagger|b \rangle$, as shown in the lower panel.}\label{fig:ineq}
\end{figure}

The left and right representations established on the sector $\mathcal H_{B_1\sqcup B_2}$ are denoted by $\hat A^{B_1\sqcup B_2}_L$ and $\hat A^{B_1\sqcup B_2}_R$.  The adjoint operation defined by $\mathcal H_{B_1\sqcup B_2}$ then satisfies  $(\hat a_L)^\dagger=\widehat{( a^\star)}_L$, and similarly for the right algebra. Finally, it is clear that operators in $\hat A^{B_1\sqcup B_2}_L$ commute with those in $\hat A^{B_1\sqcup B_2}_R$.

The remaining analysis of \cite{Colafranceschi:2023urj} was restricted to so-called diagonal sectors of the form $\mathcal H_{B_1\sqcup B_1}$; i.e., with $B_1$ diffeomorphic to $B_2$.  In that context, left and right von Neumann algebras $\mathcal A_L^{B_1}$ and $\mathcal A_R^{B_1}$ were constructed by taking completions of the representations $\hat A_L^{{B_1}\sqcup {B_1}}$ and $\hat A_R^{{B_1}\sqcup {B_1}}$ in the weak operator topology. The adjoint operation again acts as an involution on these von Neumann algebras.

A key point was then that the above trace $tr$ can be extended to positive elements of the von Neumann algebras $\mathcal A_L^{B_1}$ and $\mathcal A_R^{B_1}$ by taking it to be defined by 
\begin{equation}
\label{eq:trvN}
    tr(a) : = \lim_{\beta \downarrow 0} \langle C_\beta |a | C_\beta \rangle,
\end{equation}
where $C_\beta$ is an appropriate cylinder\footnote{Ref \cite{Colafranceschi:2023urj} instead used so-called normalized cylinders $\tilde C_\beta$, but this is unnecessary since the appendix of \cite{Colafranceschi:2023urj} shows that the appropriate norm approaches $1$ as $\beta \rightarrow 0$.} of length $\beta$.  
The result is faithful, normal, and semifinite.  In addition, it continues to satisfy the trace inequality \eqref{eq:trIn}, as well as related inequalities derived using larger numbers of boundaries.  Together, these results require $tr(P)$ to be a non-negative integer for any projection $P$ with finite trace.  Each of the algebras $\mathcal A^{B_1}_L, \mathcal A^{B_1}_R$ must thus be a direct sum of type I factors.  Furthermore, the Hilbert space $\mathcal H_{{B_1}\sqcup {B_1}}$ is a direct sum of Hilbert spaces $\mathcal H_{{B_1}\sqcup {B_1}}^\mu$ that factorize as  $\mathcal H_{{B_1}\sqcup {B_1}}^\mu =\mathcal H_{{B_1}\sqcup {B_1}, L}^\mu \otimes \mathcal H_{{B_1}\sqcup {B_1}, R}^\mu$ with $\mathcal H_{{B_1}\sqcup {B_1}, L}^\mu$ canonically isomorphic to $\mathcal H_{{B_1}\sqcup {B_1}, R}^\mu$ up to an overall phase.  Finally, it was shown that such Hilbert-space factors $\mathcal H_{{B_1}\sqcup {B_1}, L}^\mu$, $H_{{B_1}\sqcup {B_1},R}^\mu$ could be supplemented with finite dimensional Hilbert spaces ${\mathbb C}^{n_\mu}$ such that the trace $\widetilde{Tr}$ defined by summing diagonal matrix elements of operators on $\undertilde {\mathcal H}_{B_1} := \bigoplus_\mu \mathcal H_{{B_1}\sqcup {B_1}, L}^\mu\otimes {\mathbb C}^{n_\mu}= \bigoplus_\mu \mathcal H_{{B_1}\sqcup {B_1},R}^\mu\otimes {\mathbb C}^{n_\mu}$ coincides on $\mathcal A_L^{B_1}$ and $\mathcal A_R^{B_1}$ with the trace $tr$ defined above. Since $\mathcal H_{B_1 \sqcup B_1} \subset \undertilde {\mathcal H}_{B_1} \otimes \undertilde {\mathcal H}_{B_1},$ this provided a Hilbert space interpretation of the  entropy described by the gravitational replica trick.  And  by the argument of \cite{Lewkowycz:2013nqa}, this entropy is well approximated by the Ryu-Takayanagi entropy \cite{Ryu:2006bv,Ryu:2006ef} when the bulk theory admits an appropriate semiclassical limit.

\section{Off-diagonal representations from the diagonal representation}
\label{sec:alg}

The main goal of this paper is to generalize the above results to off-diagonal sectors ${\cal H}_{B_L \sqcup B_R}$ with   $B_L\ne B_R$.  We perform the first steps of that analysis in this section, focusing on the surface algebras  $A^{B_L}_L$ and $A^{B_R}_R$ and  their
representations $\hat A_L^{B_L\sqcup B_R}$ and $\hat A_R^{B_L\sqcup B_R}$  on ${\cal H}_{B_L \sqcup B_R}$.  In particular, we will show that any off-diagonal representation $\hat A_L^{B_L\sqcup B_R}$ can be identified with a quotient of $\hat A_L^{B_L\sqcup B_L}$.

The construction of these objects with $B_L \neq B_R$ was already given in \cite{Colafranceschi:2023urj} and was reviewed in section \ref{sec:rev}.  We may thus  proceed rapidly.
The surface algebras 
 $A^{B_1}_L$ and $A^{B_2}_R$ were in fact defined in section \ref{sec:rev} using only properties that are intrinsic to the spaces of surfaces $Y^d_{B_1 \sqcup B_1}, Y^d_{B_2 \sqcup B_2}$, without mention of any Hilbert space sector.  We thus need only set $B_1 = B_L$ and $B_2 = B_R$ to obtain surface algebras $A^{B_L}_L$ and $A^{B_R}_R$ which are identical to those used in the diagonal context. 
 
We will show below that the representation of $A^{B_L}_L$ on any sector ${\cal H}_{B_L \sqcup B_R}$ is always a quotient of the representation on the diagonal sector ${\cal H}_{B_L \sqcup B_L}$, and similarly for the right surface algebra.  This statement is equivalent to the claim that, if $n$ lies in the diagonal null space ${\mathfrak{N}}_L^{B_L \sqcup B_L}$ of elements of $A^{B_L}_L$  that annihilate all states in the diagonal sector
${\cal H}_{B_L \sqcup B_L}$, then $n$ must also annihilate all states in any non-diagonal sector ${\cal H}_{B_L \sqcup B_R}$. 

To streamline our notation for boundaries,  we now introduce the shorthand $LR=B_L\sqcup B_R$,  $LL=B_L\sqcup B_L$, and  $RR=B_R\sqcup B_R$. 
We will in particular write  $\hat A^{LR}_L = \hat A_L^{B_L\sqcup B_R}$ and $\hat A^{LR}_R= \hat A_R^{B_L\sqcup B_R}$.
As in the diagonal case, the adjoint operation defined by $\mathcal H_{B_L\sqcup B_R}$ satisfies  $(\hat a_L)^\dagger=\widehat{( a^\star)}_L$, and similarly for the right algebra. It is also again clear that operators in $\hat A^{LR}_L$ commute with those in $\hat A^{LR}_R$.

The representation $\hat A_L^{LR}$ is a faithful representation of the quotient algebra $A_L^{B_L}/{\mathfrak{N}}_L^{LR}$, where ${\mathfrak{N}}_L^{LR} = {\mathfrak{N}}_L^{B_L\sqcup B_R}$ is the null space consisting of elements whose operator representations annihilate the entire sector $\mathcal H_{B_L\sqcup B_R}$.  We now  make the following claim:
\begin{theorem}
\label{thm:null}
For any $B_R$, the null space ${\mathfrak{N}}_L^{B_L\sqcup B_R}$ contains the
diagonal null space  ${\mathfrak{N}}_L^{B_L\sqcup B_L}$:
    \begin{equation}
    \label{eq:null}
        {\mathfrak{N}}_L^{B_L\sqcup B_R}\supseteq {\mathfrak{N}}_L^{B_L\sqcup B_L}.
    \end{equation}
\end{theorem}
Here, and throughout the rest of this work, we could of course also discuss the corresponding properties of algebras that act on the right boundary (which here would give ${\mathfrak{N}}_R^{B_L\sqcup B_R}\supseteq {\mathfrak{N}}_R^{B_R\sqcup B_R}$).  For simplicity, we will generally refrain from doing so explicitly, though in all cases such properties clearly follow from analogous arguments.

\begin{figure}[htbp]
\hspace{1cm}\includegraphics[width=\textwidth]{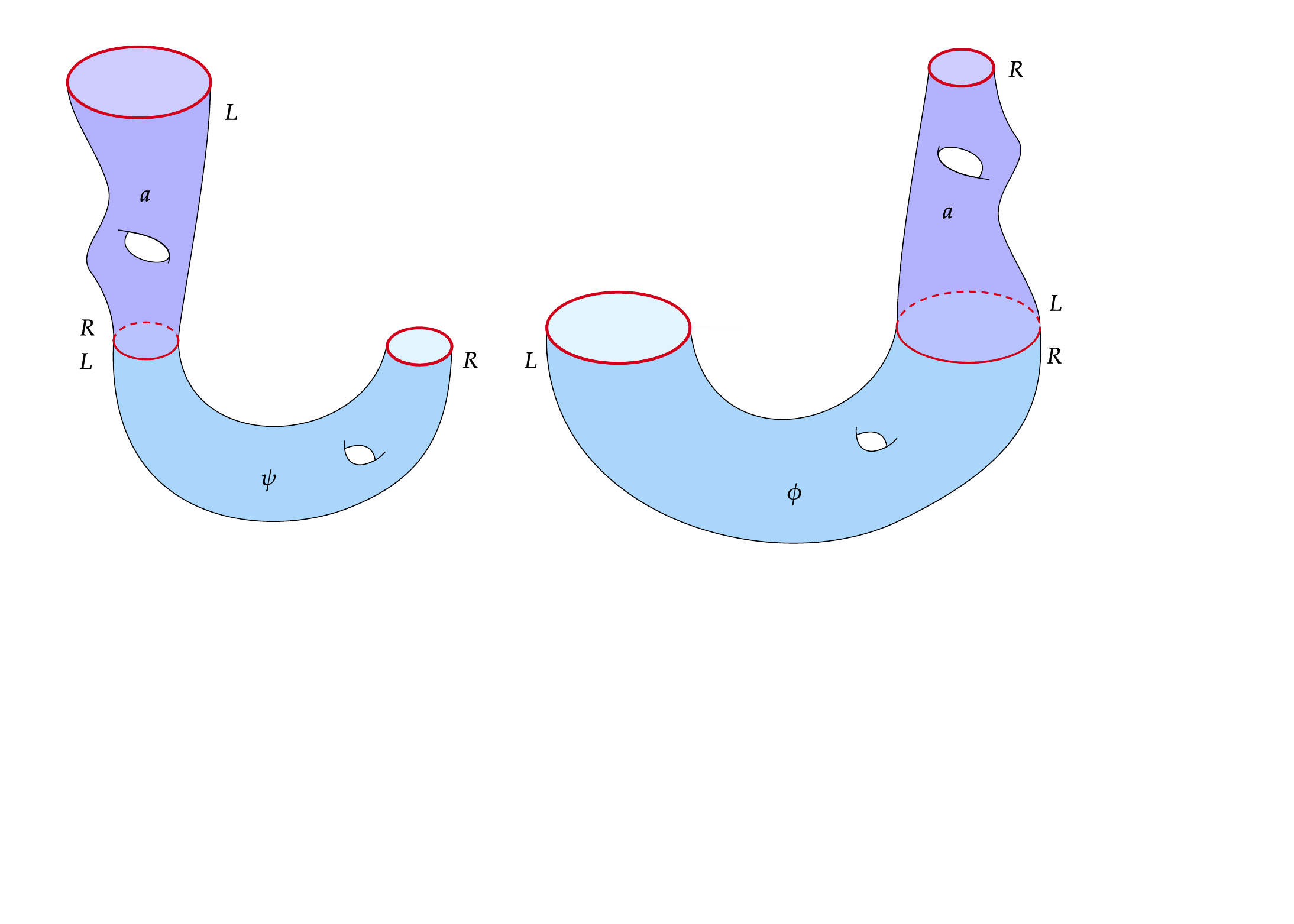}
\vspace{-5cm}
\caption{For surfaces $a\in Y^d_{LR}$, $\psi\in Y^d_{LL}$, $\phi\in Y^d_{RR}$, the left panel shows the action of $\hat a_L$ on $\ket \psi\in H_{RR}$ while the right panel shows the action of $\hat a_R$ on $\ket \phi\in H_{LL}$ as defined in equation~\eqref{eq:glue}.  Both surfaces obtained belong to $Y^d_{LR}$.\label{fig:glue}} 
\end{figure}

It is useful to introduce an additional set of bounded operators before we prove claim \ref{thm:null}. For any surface $a\in Y^d_{LR}$, we can define an operator $\hat a_R: H_{LL}\to H_{LR}$ via the usual gluing of surfaces; see figure~\ref{fig:glue}.  This is just the analogue of the action of the surface algebra $A^{B_L}_R$ on $H_{LL}$, but where the two boundaries of $a$ are now allowed to be  different (so that the set of such objects no longer forms a natural algebra, and so that $\hat a_R |\phi\rangle$ lives in a different Hilbert space sector than the original state $|\phi\rangle$).   As usual, the inequality \eqref{eq:bounded} implies that 
$\hat a_R$ maps zero-norm states to zero-norm states and so yields a well-defined operator on the image of $H_{LL}$ in $\mathcal{H}_{LL}$. The  action of the resulting operator is naturally written in the form
\begin{equation}
\label{eq:glue}
    \hat a_R\ket \phi =\ket {\phi a}\quad \forall \ket \phi \in H_{LL},
\end{equation}
where $\phi a$ denotes the surface formed by gluing the right boundary of $\phi \in Y^d_{LL}$ to the left boundary of $a \in Y^d_{LR}$.
The subscript $R$ on $ \hat a_R$ denotes the fact that this operator acts on the right boundary of $\phi$.  We  also define the analogous left operator $\hat a_L: H_{RR}\to H_{LR}$ for which
\begin{equation}
\label{eq:Lglue}
    \hat a_L\ket \psi =\ket {a \psi }\quad \forall \ket \psi \in H_{RR},
\end{equation}

Now, as described in section \ref{sec:rev}, surfaces $a$ define states $\ket a$ in the dense subspace $\mathcal D_{LR} \subset \mathcal{H}_{LR}$.  It will be useful below to think about defining the operator $\hat a_R$ directly from such $\ket a$.  To do so, we simply choose an arbitrary representative element $a\in H_{LR}$ of the equivalence class defined by $\ket a$.  This $a$ is a finite linear combination of surfaces in $Y^d_{LR}$ to which we can extend our definition of $\hat a_R$ by linearity.  We then need only observe that if two representatives $a_1, a_2$ differ by some zero-norm surface $N \in {\cal N}_{LR}$, then for any $\phi$ the norm of $\ket {\phi N}$ is given by
\begin{equation}
\langle \phi N | \phi N \rangle = \langle N| \hat \phi _L^\dagger \hat \phi _L|N\rangle \le tr(\phi ^\star \phi ) \langle N| N \rangle =0,
\end{equation}
where the inequality as usual follows from \eqref{eq:bounded}.
Thus the operators $\hat a_{1,R}$ and $\hat a_{2,R}$ are identical so that $\hat a_R$ is fully defined by the choice of the state $\ket a \in \mathcal D_{LR} \subset \mathcal{H}_{LR}$.  

For later use, we denote the associated map on $\mathcal D_{LR}$ by $\Psi_R$, and we use $\Psi_L$ for the corresponding left operator.  In particular, we will use the notation
\begin{equation}
\Psi_R(|a\rangle) := \hat a_R, \ \ \ \Psi_L(|a\rangle) := \hat a_L,
\end{equation}
where we remind the reader that $\hat a_R$ maps $H_{LL}$ to $H_{LR}$ while $\hat a_L$ maps $H_{RR}$ to $H_{LR}$.

To discuss the adjoints of  $\hat a_R$ and  $\hat a_L$, we first extend the definition of the $\star$ operation  to off-diagonal surfaces:
\begin{definition}
\label{def:inv}
    For any $a\in Y^d_{B_1\sqcup B_2}$, we define $a^\star \in Y^d_{B_2\sqcup B_1}$ to be the  source manifold-with-boundary obtained from $a$ by complex-conjugating all sources, relabeling the left boundary $B_1$ of $a$ as the right boundary of $a^\star$, and similarly  relabeling the right boundary $B_2$ of $a$ as the left boundary of $a^\star$.  As shown in figure~\ref{fig:inn}, this definition then satisfies the  relation
    \begin{equation}
    \label{eq:trprod}
        \braket{a|b}=\zeta(M(a^\star b))=tr(a^\star b)=tr(ba^\star)\quad \forall a,b\in H_{B_1\sqcup B_2},
    \end{equation}
    where the first trace acts on  $ a^\star b \in \underline Y^d_{B_2\sqcup B_2}$ and the second trace acts on  $ba^\star \in\underline Y^d_{B_1\sqcup B_1}$.
\end{definition}

\begin{figure}[htbp]
\centering
\begin{minipage}[b]{0.5\textwidth}
  \centering
  \includegraphics[clip, trim=0.5cm 15cm 0.5cm 10cm, width=\linewidth]{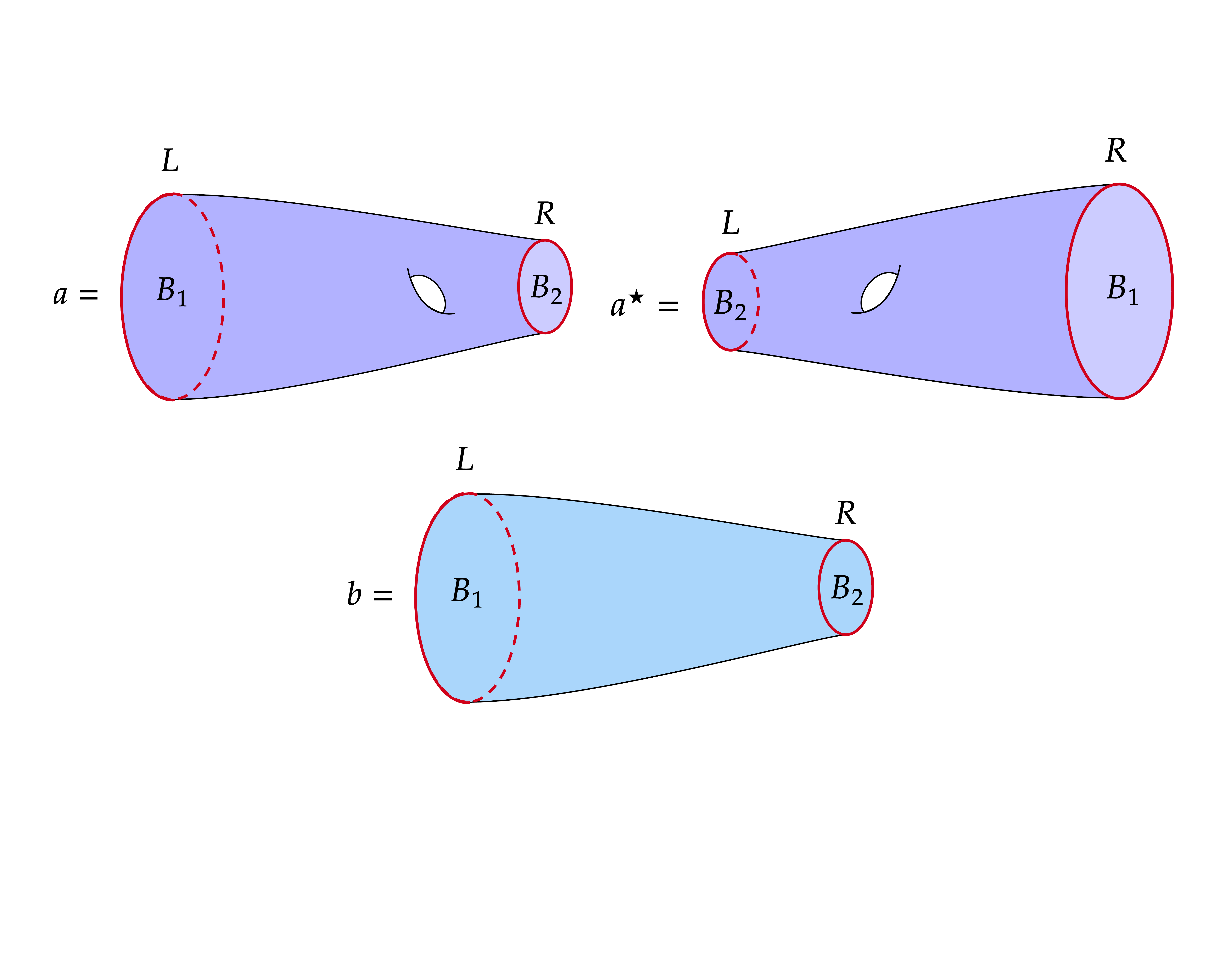}
\end{minipage}
\begin{minipage}[b]{0.4\textwidth}
  \centering
  \includegraphics[width=\linewidth]{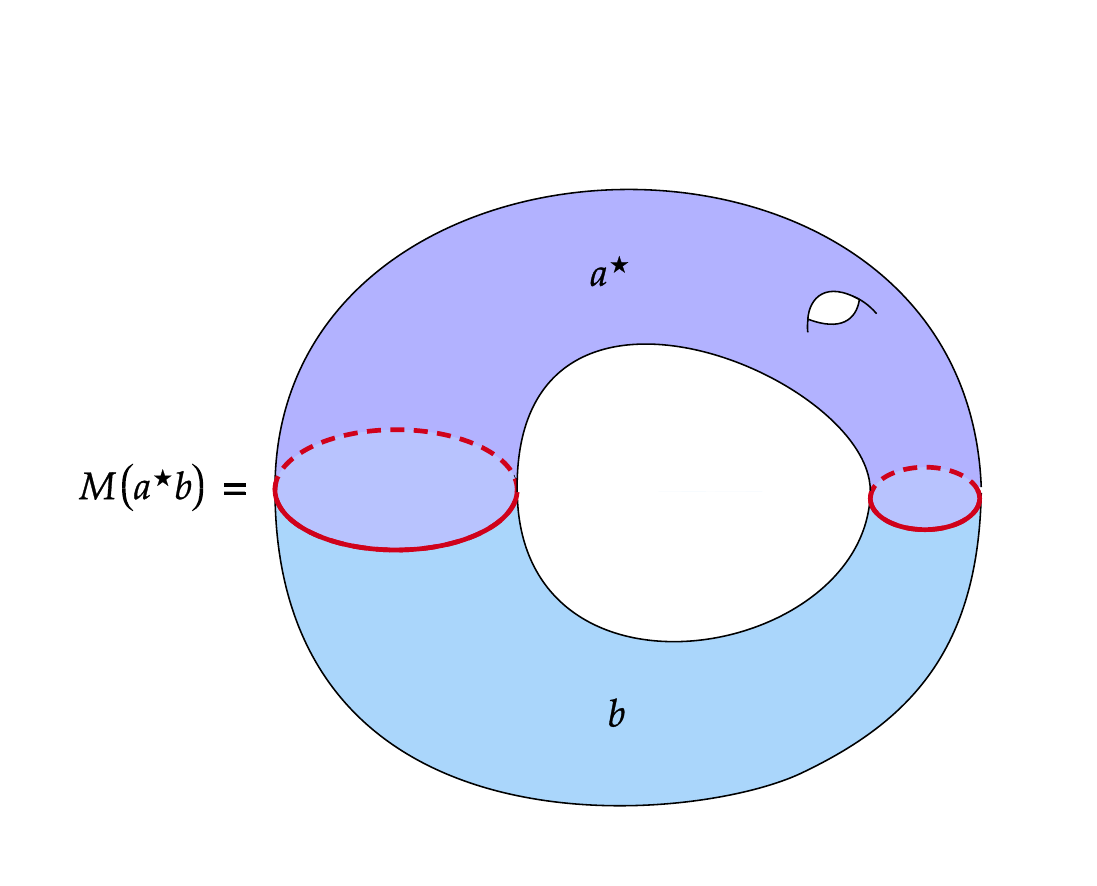}
\end{minipage}
\caption{
 For surfaces $a,b\in Y^d_{B_1\sqcup B_2}$, we can construct $a^\star$ and compute $tr(a^\star b)$ and $tr(ba^\star)$.  Both are given by evaluating $\zeta$ on the closed surface shown at right, as is $\langle a |b \rangle$  This observation yields \eqref{eq:trprod}. \label{fig:inn}}
\end{figure}

The adjoint $(\hat a_R)^\dagger:H_{LR}\to H_{LL}$ of $(\hat a_R)$ is associated with $a^\star$ in the sense that it satisfies
\begin{equation}
\label{eq:adjpsiR}
    (\hat a_R)^\dagger\ket \phi =\ket {\phi a^\star}\quad \forall \ket \phi \in H_{LR}.
\end{equation}
The result \eqref{eq:adjpsiR} can be verified by writing
\begin{equation}
\label{eq:adj}
\langle \phi |\psi a \rangle = tr(\phi ^\star \psi  a) = tr(a \phi ^\star \psi  ) = tr\left((\phi  a^\star)^\star \psi  \right) = \langle \phi  a^\star | \psi  \rangle.\end{equation}

As usual, we can use the trace inequality \eqref{eq:trIn} to show that $\hat a_R$ in fact defines a bounded operator that maps the Hilbert space $\mathcal{H}_{LL}$ to the Hilbert space $\mathcal{H}_{LR}$.    In this case we consider again $|\phi \rangle \in H_{LL}$.  After setting $b=\phi$, the right analogue of 
\eqref{eq:bounded} yields
\begin{equation}
\label{eq:bound}
\langle \phi  | \hat a_R^\dagger \hat a_R |\phi  \rangle \le \langle \phi | \phi \rangle \langle a | a \rangle.
\end{equation}
Thus $\hat a_R$ annihilates null states in $H_{LL}$ and defines a bounded operator on the dense subspace ${\cal D}_{LL} = H_{LL}/{{\mathcal{N}}}_{LL} \subset \mathcal{H}_{LL}$.  A unique bounded extension to $\mathcal{H}_{LL}$ then follows.

As a result, the operation $\Psi_R: |a \rangle \mapsto \hat a_R$ defines a linear map from $\mathcal D_{LR}\subset \mathcal H_{LR}$ to the space $\mathcal B(\mathcal H_{LL},\mathcal H_{LR})$ of bounded operators from $\mathcal H_{LL}$ to $\mathcal H_{LR}$.   We will extend this map to the entire Hilbert space in section~\ref{sec:psi}. In the diagonal context $B_L=B_R$, the map $\Psi_R: H_{LL}\to \hat A_R^{LL}$ coincides with the representation of the right surface algebra under the natural identification of $H_{LL}$  with $A_R^{B_L}$.

Returning to  claim  \ref{thm:null}, we need to show that any element $a$ of the surface algebra $A_L^{B_L}$ that is represented by the zero operator on the diagonal Hilbert space $\mathcal H_{LL}$ is also represented by zero on any $\mathcal H_{LR}$.  And since our operators are bounded, it in fact suffices to show that the representation $\hat a_L\in \hat A^{LR}_L$ of $a$ on $\mathcal H_{LR}$ annihilates every state $|\psi\rangle$ in the dense subspace $\mathcal D_{LR}$.

\begin{figure}[htbp]
\centering
\includegraphics[clip, trim=0.5cm 1.5cm 0.5cm 0cm, width=0.8\textwidth]{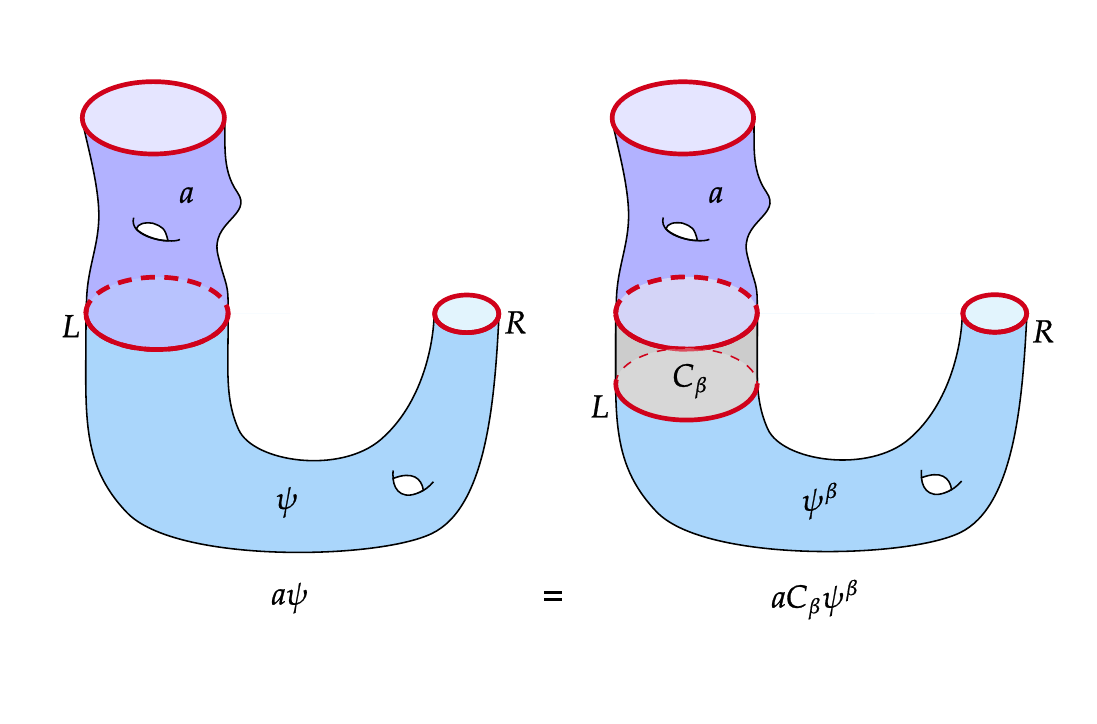}
\caption{  For rimmed surfaces $\psi\in Y^d_{LR}$ and $a\in Y^d_{LL}$, we define the action of $\hat a_L$ on $\ket \psi$ as $\ket {a\psi}$. As shown in the right panel, we can also separate out a cylinder $C_\beta$ from the left rim of $\psi$ and rewrite the surface $a\psi$ as $aC_\beta \psi^\beta$, where $\psi^\beta\in Y^d_{LR}$. \label{fig:phi}}
\end{figure}

\begin{proof}
Let us consider a rimmed surface $\psi \in Y^d_{B_L \sqcup B_R}$ and the corresponding state $|\psi\rangle \in \mathcal D_{LR}$.  The rim at the left boundary requires there to be a neighborhood of the left boundary whose closure coincides with some cylinder $C_\beta$ for some $\beta >0$. As shown in figure \ref{fig:phi}  we may thus write the surface $a\psi$ as $aC_\beta\psi^\beta$ for some $\psi^\beta \in Y^d_{B_L \sqcup B_R}$.  This observation yields the relations
\begin{equation}
\label{eq:null3}
    \hat a_L\ket \psi=\hat a_L\ket {C_\beta\psi^\beta}=\ket {aC_\beta \psi^\beta}=\hat \psi^\beta_R\ket {a C_\beta}=\hat\psi_R^\beta \hat a_L^{LL}\ket {C_\beta}\quad \forall \ket\psi\in \mathcal D_{LR},
\end{equation}
where $\hat a_L^{LL}$ is the representation of $a$ on the diagonal sector $\mathcal H_{LL}$. Thus $\hat a_L^{LL}=0$ requires $\hat a_L\ket \psi=0$.  Since $\mathcal D_{LR}$ is the linear span of $|\psi\rangle$ with $\psi \in Y^d_{B_L \sqcup B_R}$, it follows that $\hat a_L$ annihilates $\mathcal D_{LR}$. This establishes ${\mathfrak{N}}_L^{LL}\subseteq {\mathfrak{N}}_L^{LR}$ as claimed above.
\end{proof}

As a result, every left representation $\hat A^{LR}_L$ of the surface algebra $A_L^{B_L}$ is also a representation of the diagonal representation $\hat A^{LL}_L$.   This is equivalent to the statement that there is a surjective $*$-homomorphism $\hat \Phi_{LR}$ from $\hat A^{LL}_L$ to any $\hat A^{LR}_L$.  It follows that $\hat A^{LR}_L$ may be identified with the quotient of $\hat A^{LL}_L$ by an appropriate kernel, and thus that the diagonal representation contains all information about all representations $\hat A^{LR}_L$.  We will show below that analogous statements continue to hold for the von Neumann algebra completions.

\section{Completions and von Neumann algebras}
\label{sec:hom}

We saw above that, for a given surface algebra $A_L^{B_L}$,  the diagonal representation $\hat A^{LL}_L$ acts as universal `covering algebra' for any representation $\hat A^{LR}_L$ in the sense that there is a $*$-homomorphsim $\hat \Phi_{LR}$ from
$\hat A^{LL}_L$ to $\hat A^{LR}_L$.   Furthermore, as reviewed in section \ref{sec:rev}, completing the diagonal representation $\hat A^{LL}_L$   yields an associated  von Neumann algebra $\mathcal A_L^{B_L}$.  

We will now generalize this construction to define an off-diagonal von Neumann algebra $\mathcal A_L^{LR}=\mathcal A_L^{B_L\sqcup B_R}$ by completing 
$\mathcal A_L^{LR}$ with respect to the weak operator topology on $\mathcal B(\mathcal H_{B_L \sqcup B_R})$.  It is thus natural to ask if the results of section \ref{sec:alg} extend to such completions.  To make the notation more uniform between the diagonal and off-diagonal contexts, we will henceforth use the symbol $\mathcal A_L^{LL} = \mathcal A_L^{B_L \sqcup B_L}$ for the von Neumann algebra $\mathcal A_L^{B_L}$ defined using the weak operator topology on the diagonal Hilbert space $\mathcal H_{B_L \sqcup B_L}.$

The goal of this section is to show that this is indeed the case.  The subtlety, however, is that the weak operator topology is defined by the Hilbert space on which the operators act.  Thus, a priori, two completions might be very different even if we begin with isomorphic representations.  It will turn out, however, that our axioms in fact require there to be a simple relation between $\mathcal A_L^{LR}$ and $\mathcal A_L^{LL}$.

As a result, we will be able to extend the surjective *-homomorphism $\hat \Phi_{LR}: \hat A_L^{LL}\to \hat A_L^{LR}$ to 
a surjective *-homomorphism $\Phi_{LR}: \mathcal A_L^{LL}\to \mathcal A_L^{LR}$ that is continuous with respect to the weak operator topology.   Note that the extended $\Phi_{LR}$ is written without a hat decoration$\hat{\quad}$.  Extending $\hat \Phi_{LR}$ means in particular that, given any $a \in \mathcal A_L^{LL}$, we must construct an appropriate bounded operator $\Phi_{LR}(a)$ on the Hilbert space $\mathcal H_{LR}$.  This first step will be accomplished in section \ref{sec:extend}, which also verifies that the extension defines a $*$-homomorphism.  After an aside to introduce some useful technology in section \ref{sec:psi}, the desired continuity will then be established in section \ref{sec:con}.

\subsection{Extending our homomorphism to the von Neumann algebra $\mathcal A_L^{LL}$}
\label{sec:extend}

For each $a$ in the von Neumann algebra $\mathcal A_L^{LL}$, 
we will first define the desired $\Phi_{LR}(a)$ as an operator on the dense domain $\mathcal D_{LR}$.  Any state $|\psi \rangle \in \mathcal D_{LR}$ is a finite linear combination 
\begin{equation}
    \label{eq:span}
\sum_{i=1}^N c_i |\psi_i \rangle
\end{equation}
for states $|\psi_i\rangle$ defined by rimmed surfaces  $\psi_i$ with given right and left boundaries.  Since each such surface will have a neighborhood of its left boundary that coincides (up to closure) with some cylinder $C_\epsilon$ for $\epsilon >0$, by choosing $\epsilon >0$ sufficiently small we can write $\psi_i = C_\epsilon \psi_i^\epsilon$ for all $i$; see again figure \ref{fig:phi}.  Defining the state $\ket {\psi^\epsilon}  =  \sum_{i=1}^N c_i |\psi_i^\epsilon \rangle$, we may use the associated operator $\hat \psi^\epsilon_R := \Psi_R(|\psi^\epsilon \rangle)$ from $\mathcal H_{LL}$ to $\mathcal H_{LR}$ given by \eqref{eq:glue}.  We thus make the following definition for any $a\in \mathcal A_L^{LL}:$

\begin{equation}
\label{eq:dense}
    \forall \ket \psi \in \mathcal D_{LR}\quad {\rm we \ define} \quad \Phi_{LR}(a) \ket {\psi} := \hat \psi_R^\epsilon a |C_\epsilon\rangle\quad {\rm for \ small \ enough}\ \epsilon.
\end{equation}
It is manifest that $\Phi_{LR}(a) = \hat \Phi_{LR}(a)$ for $a\in \hat A^{LR}_L$, since then $\hat \Phi_{LR}(a)|\psi\rangle =|a\psi\rangle$ and \eqref{eq:dense} reduces to \eqref{eq:glue}. Furthermore, for all $a$ in the von Neumann algebra $\mathcal A_L^{LL}$, we will show below that the definition \eqref{eq:dense} is independent of $\epsilon$ for small enough $\epsilon.$
 It is then also manifest that  $\Phi_{LR}(ab) = \Phi_{LR}(a)\Phi_{LR}(b)$ and $\Phi_{LR}(\alpha a+ \beta b) = \alpha  \Phi_{LR}(a) + \beta \Phi_{LR}(b)$ for $\alpha, \beta \in \mathbb{C}$, so that $\Phi_{LR}$ is a homomorphism.

We now establish several claims regarding the definition \eqref{eq:dense}.

\begin{theorem}
\label{thm:indep}
    The definition \eqref{eq:dense} is independent of $\epsilon$ so long as $C_\epsilon$ is a cylinder common to all rimmed surfaces $\psi_i$ appearing in \eqref{eq:span}. 
\end{theorem}

\begin{proof}
    For $a\in \hat A_L^{LL}$, this claim follows from the observation that  $\Phi_{LR}(a) = \hat \Phi_{LR}(a)$. Furthermore, a general $a$ in the von Neumann algebra $\mathcal A_L^{LL}$ can be written as the limit of a net of operators $\{\hat a_\alpha\}\subset \hat A_L^{LL}$ that converges to $a$ in the {\it strong} operator topology.\footnote{For convex sets of bounded operators, the closure taken in the weak operator topology agrees with that taken in the strong operator topology; see e.g. theorem 5.1.12 in~\cite{KR1}. Here we regard the von Neumann $\mathcal A_L^{LL}$ as the closure of $\hat A^{LL}_L$ in the strong operator topology, which then requires that we include the limits of all strongly-convergent nets.} As a result,  since Hilbert spaces are metrizable \footnote{Here we used an extended sequential property of metrizable spaces. Given two convergent nets in a metrizable space $\{\ket {\psi_\alpha}\}\to \ket \psi$ and $\{\ket {\psi'_\alpha}\}\to \ket{\psi'}$ with a common index set $\mathcal J$, there exist subsequences $\{\ket{\psi_{\alpha_n}}\}\to \ket \psi$ and $\{\ket{\psi'_{\alpha_n}}\}\to \ket {\psi'}$ 
    with common indices $\alpha_n\in \mathcal J$ that
    converge to the same limit point.  We presume the argument for this result` is standard, but reader's seeking an explict reference can consult the discussion around (3.42) in \cite{Colafranceschi:2023urj}.}, given any
    two common rims $C_\epsilon$ and $C_\delta$
    there must be a {\it sequence} $\{\hat a_n\}\subset \hat A_L^{LL}$ such that the sequence $\{\hat a_n |C_\epsilon\rangle\}$ converges in norm to $a|C_\epsilon\rangle$ while the sequence 
    $\{\hat a_n |C_\delta\rangle\}$ converges in norm to $a|C_\delta\rangle$.  Since $\hat \psi_R^\epsilon$ and 
   $\hat \psi_R^\delta$ are bounded, we similarly find the limits
\begin{eqnarray}   
\label{eq:lims}
   \hat \psi_R^\epsilon \hat a_n |C_\epsilon\rangle \rightarrow 
   \hat \psi_R^\epsilon a |C_\epsilon\rangle \nonumber\\
   {\rm and} \ \ \    \hat \psi_R^\delta \hat a_n|C_\delta\rangle \rightarrow 
   \hat \psi_R^\delta a |C_\delta\rangle. 
   \end{eqnarray}
But since $\hat a_n \in \hat{A}^{LL}_L$ we have   $\hat \psi_R^\epsilon \hat a_n |C_\epsilon\rangle = \hat \Phi_{LR}(\hat a_n) |\psi\rangle = \hat \psi_R^\delta \hat a_n |C_\delta\rangle$ as noted in the opening sentence of this proof.  Thus  \eqref{eq:lims} implies 
$\hat \psi_R^\epsilon a |C_\epsilon\rangle = 
   \hat \psi_R^\delta a |C_\delta\rangle$ as claimed.

\end{proof}

\begin{theorem}
\label{thm:bound}
    For any $a \in \mathcal A_L^{LL}$,  our $\Phi_{LR}(a)$ is a bounded operator on $\mathcal D_{LR}$ (with operator norm no larger than the norm $||a||$ of $a$).  It thus admits a unique bounded extension to the entire Hilbert space $\mathcal H_{LR}$.
\end{theorem}

\begin{proof}
We begin by computing the norm of  \eqref{eq:dense}.  Since $|\psi\rangle \in \mathcal D_{LR},$ we have
    \begin{equation}
    \label{eq:norm}
        |\Phi_{LR}(a) \ket\psi|^2=|\hat \psi^\epsilon_R a \ket{C_\epsilon}|^2=\bra{C_\epsilon} a^\dagger \hat \psi^{\epsilon,\dagger}_R \hat \psi^\epsilon_R a\ket{C_\epsilon}.
    \end{equation}

Recall now that $\hat \psi^\epsilon_R$ is a bounded operator from $\mathcal H_{LL}$ to $\mathcal H_{LR}$.  As a result, $\hat \psi^{\epsilon,\dagger}_R \hat \psi^\epsilon_R$ 
is a bounded operator on $\mathcal H_{LL}$.  Furthermore, since $|\psi\rangle$ is of the form \eqref{eq:span}, the operator $\hat \psi^{\epsilon,\dagger}_R \hat \psi^\epsilon_R$  is a finite linear combination of the operators $\hat \psi^{\epsilon,\dagger}_{i,R} \hat \psi^\epsilon_{j,R}$, each of which lies in the right representation $\hat A^{LL}_R$ of the surface algebra $A_R^{B_L}$.  Thus $\hat \psi^{\epsilon,\dagger}_R \hat \psi^\epsilon_R \in \hat A^{LL}_R \subset \mathcal A_R^{LL}$, so that its (unique) positive square root $|\hat \psi^\epsilon_R|$ also lies in the von Neumann algebra $\mathcal A_R^{LL}$.  We may then use the fact that $a$ and $a^\dagger$ are in the left von Neumann algebra $\mathcal A^{LL}_L$ to conclude that they commute with any operator in the right von Neumann algebra $\mathcal A_R^{LL}$, and in particular with $|\hat \psi^\epsilon_R|$.  This observation allows us to write
\begin{eqnarray}
\label{eq:bound2}
    \bra{C_\epsilon} a^\dagger \hat \psi^{\epsilon,\dagger}_R \hat \psi^\epsilon_R a\ket{C_\epsilon} &=&
        \bra{C_\epsilon} a^\dagger |\hat \psi^{\epsilon}_R|^2 a\ket{C_\epsilon}
    =    \bra{C_\epsilon} |\hat \psi^{\epsilon}_R| \, a^\dagger  a \, |\hat \psi^{\epsilon}_R|\ket{C_\epsilon} \cr
    &\le& ||a||^2 \bra{C_\epsilon} |\hat \psi^{\epsilon}_R|^2 \ket{C_\epsilon}
    = ||a||^2    \bra{C_\epsilon}  \hat \psi^{\epsilon,\dagger}_R \hat \psi^\epsilon_R \ket{C_\epsilon} 
    =   ||a||^2  \bra{\psi} \psi \rangle,
\end{eqnarray}
where $||a||$ is the operator norm of $a$.  Thus $\Phi_{LR}(a)$ is bounded as claimed and, in particular, its norm can be no larger than the norm of $a$.
\end{proof}

For each $a \in \mathcal A_L^{LL}$, we may now extend the domain of $\Phi_{LR}(a)$ to all of $\mathcal H_{LR}$ by continuity to define $\Phi_{LR}$ as a linear map from $\mathcal A_L^{LL}$ to the space $\mathcal B(\mathcal H_{LR})$ of bounded linear operators on $\mathcal H_{LR}$.  
Since bounded operators are determined by their matrix elements on the dense subspace $\mathcal D_{LR}$, this extension is again a homomorphism.  One can also quickly use \eqref{eq:dense} (and the fact that $a \in \mathcal A_L^{LL}$ commutes with  $\hat \psi^{\epsilon,\dagger}_{1,R} \hat \psi^\epsilon_{2,R} 
\in \mathcal A_R^{LL}$ for any $|\psi_1\rangle, |\psi_2\rangle \in \mathcal D_{LR}$) to show that $\Phi_{LR}$ is a $*$-homomorphism, meaning that $\Phi_{LR}(a^\dagger) = \left(\Phi_{LR}(a)\right)^\dagger$.  We will also mention that we will find in section \ref{sec:con} that $\Phi_{LR}(a)$ is continuous with respect to the weak operator topology.  Since we know that $\Phi_{LR}$ maps $\hat A_L^{LL} \subset \mathcal A_L^{LL}$ onto $\hat A_L^{LR}$, and since the $\mathcal A_L^{LR}$ is the closure of  $\hat A_L^{LR}$ in the weak operator topology, it will then follow that $\Phi_{LR}: \mathcal A_L^{LL} \rightarrow A_L^{LR}$ is a surjective $*$-homomorphism.  This is precisely the analogue of our result from section \ref{sec:alg} at the level of the von Neumann algebras $A_L^{LL}, A_L^{LR}$.

\subsection{Vectors in an off-diagonal sector as intertwining operators}
\label{sec:psi}

Before proceeding to the proof of continuity in section \ref{sec:con}, it will be useful to first derive some additional properties of the map $\Psi_R$.
Recall that section~\ref{sec:alg} defined $\Psi_R$ as  a linear map from $\mathcal D_{LR}$ to $\mathcal B(\mathcal H_{LL}, \mathcal H_{LR})$. The results of section \ref{sec:extend} will turn out to imply this $\Psi_R$ to be continuous in the strong operator topology.  In particular,  we now establish the following claim:

\begin{theorem}
\label{thm:strong}
The map $\Psi_R:\mathcal D_{LR}\to \mathcal B(\mathcal H_{LL}, \mathcal H_{LR})$ is continuous w.r.t. the norm topology on $\mathcal H_{LR}$ and the strong operator topology on $\mathcal B(\mathcal H_{LL}, \mathcal H_{LR})$. It then follows that $\Psi_R$ admits a unique continuous extension to the entire Hilbert space.
\end{theorem}

\begin{proof}
    Since $\mathcal D_{LR}$ is a metric space, the map $\Psi_R$ is continuous if and only if it acts continuously on preserves all convergent sequences. Consider then an arbitrary sequence of vectors $\{\ket{ \psi_n}\}\subset \mathcal D_{LR}$ that converges to some $|\psi\rangle \in \mathcal H_{LR}$.
    The first step is to show that the corresponding sequence of operators $\{\widehat {\psi_n}_R=\Psi_R(\ket { \psi_n})\}$ is uniformly bounded. In particular, replacing $|a\rangle$ by $|\psi_n\rangle$ in~\eqref{eq:bound} shows that $||\widehat {\psi_n}_R||$ is bounded by $\sqrt{\braket { \psi_n| \psi_n}}$. In addition, since the sequence $\{\ket { \psi_n}\}$ converges in $\mathcal H_{LR}$, we know the sequence of norms $\{\braket { \psi_n| \psi_n}\}$ is convergent and must be bounded. The sequence of operator norms $\{||\hat  \psi_n||\}$ is thus bounded as well, so that the sequence $\{\hat  \psi_n\}$ is uniformly bounded.

    Next, for any $\ket x\in \mathcal D_{LL}$ we will show that the sequence of vectors $\{\widehat {\psi_n}_R\ket x\}$ converges in $\mathcal H_{LR}$. Let us first note that \eqref{eq:glue} implies
    \begin{equation}   
    \label{eq:preINT}
\widehat {\psi_n}_R\ket x=\ket {x\psi_n}=\hat\Phi_{LR}(\hat x_L)\ket {\psi_n},
    \end{equation}
    where $x\in H_{LL}$ is any representative of the vector $\ket x$, and where $\hat x_L :=\Psi_L(\ket x)\in \hat A_L^{LL}$ is the bounded operator associated with the representation of $x\in A_L^{B_L}$ on the diagonal sector $\mathcal H_{LL}$. Since the sequence of vectors $\ket {\psi_n}$ converges to $|\psi\rangle$, and since $\hat \Phi_{LR}(\hat x_L)$ is bounded, \eqref{eq:preINT} shows that the sequence of vectors $\widehat {\psi_n}_R\ket x$ also converges to $\hat \Phi_{LR}(\hat x_L) |\psi\rangle$. 
    
    Via a standard calculation, it then follows that, given any Cauchy sequence $\{|x_n\rangle\} \subset \mathcal D_{LR}$, the sequence $\widehat {\psi_n}_R|x_n\rangle$ is also Cauchy, and that it thus converges in $\mathcal H_{LR}$ (see e.g. theorem 6 in section 15.2 of~\cite{Lax}).  This means that the sequence of operators $\widehat {\psi_n}_R$ converges in the strong operator topology to an operator $\hat \psi_R$. In particular, for any $|\psi\rangle \in \mathcal H_{LR}$, the limit yields a bounded operator $\hat \psi_R $ satisfying
\begin{equation}
\label{eq:phi}
    \forall \ket x\in \mathcal D_{LL},\quad\hat \psi_R\ket x=\hat \Phi_{LR}(\hat x_L)\ket{\psi}.
\end{equation}
We may thus define $\Psi_R(|\psi\rangle) :=\hat \psi_R$.  
\end{proof}

We will now show the extended $\Psi_R$ to satisfy the following extension of \eqref{eq:glue}.

\begin{theorem}
\label{thm:intw} $\forall \ket \psi\in \mathcal H_{LR}$, the operator $\hat \psi_R=\Psi_R(\ket \psi):\mathcal H_{LL}\to\mathcal H_{LR}$ is an intertwining operator between the von Neumann algebras $\mathcal A_L^{LL}$ and $\mathcal A_L^{LR}$. Specifically, we have
    \begin{equation}
    \label{eq:com2}
    \Phi_{LR}(a)\hat \psi_R=\hat \psi_R  a\quad \forall a\in \mathcal A^{LL}_L.
    \end{equation}
\end{theorem}
\begin{proof}
Any $\ket \psi \in \mathcal H_{LR}$ is the Hilbert space limit of a sequence $\{ \ket {\psi_n} \} \subset \mathcal D_{LR}$.  Note that for any $\epsilon >0$ we may construct the states $\ket{\phi_n} := \ket {C_\epsilon \psi_n} \in \mathcal D_{LR}$, and that $\widehat {\phi_n}_R^\epsilon = \widehat {\psi_n}_R$.  
For any $a\in \mathcal A_L^{LL}$ 
we may thus write
\begin{equation}
\Phi_{LR}(a) \widehat {\psi_n}_R |C_\epsilon  \rangle =
\Phi_{LR}(a) |C_\epsilon \psi_n \rangle =
\widehat {\psi_n}_R a |C_\epsilon \rangle,
\end{equation}
where the first equality is the definition of $\widehat {\psi_n}_R$ and the second is a direct application of 
\eqref{eq:dense}.  Recalling that $\widehat {\psi_n}_R = \Psi_R (|\psi_n\rangle)$, and that 
$\hat {\psi}_R = \Psi_R (|\psi \rangle)$, boundedness of $\Phi_{LR}(a)$ and the continuity property of Claim \ref{thm:strong} allow us to take the limit $n\rightarrow \infty$ to find 
\begin{equation}
\label{eq:intCe}
\Phi_{LR}(a) \hat {\psi}_R |C_\epsilon  \rangle =
\hat {\psi}_R a |C_\epsilon \rangle.
\end{equation}

To show that the above intertwining relation in fact holds when acting on arbitrary states in $\mathcal H_{LL}$, we simply consider any $\kappa \in \underline Y_{LL}^d$ and compute
\begin{eqnarray}
\label{eq:intkappa}
\hat {\psi}_R a | \kappa C_\epsilon \rangle
= 
\hat {\psi}_R a \hat \kappa_L |  C_\epsilon \rangle
&=&
\Phi_{LR}(a\hat \kappa_L ) \hat {\psi}_R |C_\epsilon  \rangle 
\cr
&=&
\Phi_{LR}(a) \Phi_{LR}(\hat \kappa_L ) \hat {\psi}_R |C_\epsilon  \rangle \cr
&=& \Phi_{LR}(a)  \hat {\psi}_R \hat \kappa_L |C_\epsilon  \rangle
= \Phi_{LR}(a)  \hat {\psi}_R |\kappa C_\epsilon  \rangle.
\end{eqnarray}
Here the last equality on the first line uses \eqref{eq:intCe} with the $a$ of \eqref{eq:intCe}  replaced by $a\hat \kappa_L$. We then pass to the second line using the fact that $\Phi_{LR}$ is a homomorphism as shown at the end of section \ref{sec:extend}.  The third line follows by applying \eqref{eq:intCe} with the $a$ of \eqref{eq:intCe} replaced by $\hat \kappa_L$. 
Since the operators $\hat \phi_R, a, \Phi_{LR}(a)$ are bounded, and since every state in the dense subspace $\mathcal D_{LR} \subset \mathcal H_{LR}$ is of the form $|\kappa C_\epsilon\rangle$ for some $\kappa, \epsilon$,  the general result \eqref{eq:com2} follows by taking corresponding limits of \eqref{eq:intkappa}.
\end{proof}

Before proceeding to the next section, we also wish to establish a further useful property of the
 extended map $\Psi_R$.  This property is a partial extension of \eqref{eq:trprod} associated with the fact that $\hat \psi^\dagger_R \hat \psi_R$ lies in $\mathcal A_R^{LL}$. In particular, we will find
\begin{equation}
\label{eq:trnorm}
    tr(\hat \psi^\dagger_R \hat \psi_R)=\braket{\psi|\psi}\quad \forall \ket \psi\in \mathcal H_{LR}.
\end{equation}
To see this, we will use the following two results:

\begin{theorem}
\label{thm:iden}
    For cylinder operators $\widehat{C_\beta}_L \in \hat A_L^{LL}$, the limit $\lim_{\beta\downarrow 0}\Phi_{LR}(\widehat{C_\beta}_L)$ converges in the strong operator topology to the identity 
 ${\mathds 1}\in B(\mathcal H_{LR})$.
\end{theorem}

\begin{proof}
    Since $\widehat{C_\beta}_L$ is the representation of the cylinder $C_\beta$ on $\mathcal H_{LL}$, we know from section \ref{sec:alg} that $\Phi_{LR}(\widehat{C_\beta}_L)$ is just the corresponding representation of $C_\beta$ on the off-diagonal Hilbert space $\mathcal H_{LR}$.  Furthermore, from claim \eqref{eq:bound} we have $||\Phi_{LR}(\widehat{C_\beta}_L)||\le||\widehat{C_\beta}_L||$.  And since it was shown in appendix A of \cite{Colafranceschi:2023urj} that $||\hat C_\beta||\to 1$ as $\beta\downarrow 0$, the operators  $\Phi_{LR}(\widehat{C_\beta}_L)$ are uniformly bounded at small $\beta$. We also observe that, since states in $\mathcal H_{LR}$ are determined by their inner products with states in $\mathcal D_{LR}$, for any $\ket x\in \mathcal D_{LR}$  the continuity axiom of \cite{Colafranceschi:2023urj} requires the vectors $\Phi_{LR}(\widehat{C_\beta}_L)\ket x=\ket {C_\beta x}$ to converge  to $\ket x$ in the limit $\beta\downarrow 0$.  Together, as in the proof of claim \ref{thm:strong}, these properties imply that $\Phi_{LR}(\widehat{C_\beta}_L)$ converges in the strong operator topology to ${\mathds 1}$ as $\beta \downarrow 0$; see e.g. theorem 15.2.6 of~\cite{Lax}.
\end{proof}
\begin{theorem}
\label{thm:Crep}
For any  $|\psi\rangle \in {\cal H}_{B_1 \sqcup B_2}$, we have
\begin{equation}
\label{eq:Crep`}
    |\psi \rangle  = \lim_{\beta \downarrow 0} \hat \psi_R |C_\beta \rangle
\end{equation}
for $C_\beta$ the cylinder of length $\beta$ with boundary $B_1 \sqcup B_1$.
\end{theorem}
\begin{proof}
Since any such state $|\psi\rangle$ is fully determined by its inner products with states in $\mathcal D_{B_1 \sqcup B_2}$, is clear from the continuity axiom that \eqref{eq:Crep`} holds for $|\psi \rangle \in \mathcal D_{B_1 \sqcup B_2}$.  For more general     $|\psi \rangle$, we may consider a sequence of states $|\psi_n \rangle\in \mathcal D_{B_1 \sqcup B_2}$ that converge to $|\psi\rangle.$  Continuity of $\Psi_R$ and boundedness of each operator $\widehat {\left(C_\beta\right)}_L$ then gives
\begin{eqnarray}
    \lim_{\beta \downarrow 0} \hat \psi_R |C_\beta \rangle &=& \lim_{\beta \downarrow 0} \lim_{n\rightarrow \infty} \widehat {\left(\psi_n\right)}_R |C_\beta \rangle 
    \cr &=& \lim_{\beta \downarrow 0} \lim_{n\rightarrow \infty} \hat \Phi_{LR}\left( \widehat {\left(C_\beta\right)}_L  \right)|\psi_n \rangle\cr 
    &=& \lim_{\beta \downarrow 0}  \hat \Phi_{LR} \left(\widehat {\left(C_\beta\right)}_L\right) |\psi \rangle = |\psi\rangle, 
\end{eqnarray}
where we pass from the first to the second line using \eqref{eq:phi}, and where the final step uses claim~\ref{thm:iden}.
\end{proof} 

Equation~\eqref{eq:trnorm} then follows from \eqref{eq:Crep`} by applying definition \eqref{eq:trvN} of the trace on positive elements of $\mathcal A_R^{LL}$.

\subsection{Continuity of $\Phi_{LR}$ in the weak operator topology}
\label{sec:con}

The goal of this section is to show that our $*$-homomorphism $\Phi_{LR}$ is continuous with respect to the weak operator topology (imposed on both $\mathcal B(\mathcal H_{LL})$ and $\mathcal B(\mathcal H_{LR})$). In order to do so, we first establish the following intermediate claim.

\begin{theorem}
\label{thm:diag} 
For any
$\ket \gamma\in \mathcal H_{LR}$ and any $a \in \mathcal A_L^{LL}$, we have 
\begin{equation}
    \quad\braket{\gamma|\Phi_{LR}(a)|\gamma}=\braket{\gamma_{LL}|a|\gamma_{LL}},\quad
\end{equation}
where we have defined $|\gamma_{LL} \rangle := \lim_{\beta \downarrow 0}  |\hat \gamma_R| |C_\beta \rangle \in \mathcal H_{LL}$ and where $|\hat \gamma_R|$ is the positive square root of $\hat \gamma^\dagger_R \hat \gamma_R$.
\end{theorem}

\begin{proof}
    Let us first rewrite $\braket{\gamma| \Phi_{LR}(a)|\gamma}$ by inserting the identity operator  in the form $\lim_{\beta \downarrow 0}\Phi_{LR}(\widehat{C_\beta}_L)$ established in claim~\ref{thm:iden}.  As noted in the proof of that claim, if $\widehat{C_\beta}_L$ is a cylinder operator on $\mathcal H_{LL}$, then $\hat \Phi_{LR}(\widehat{C_\beta}_L)$ is the corresponding cylinder operator on $\mathcal H_{LR}$.   We thus have 
\begin{equation}
\label{eq:weak}
\braket{\gamma|\Phi_{LR}(a)|\gamma}=
    \lim_{\beta\downarrow 0}\braket{\gamma|\Phi_{LR}(a)  \Phi_{LR}(\widehat{ C_\beta}_L)|\gamma}=    
    \lim_{\beta\downarrow 0}\braket{\gamma|\Phi(a)\hat \gamma_R|C_\beta}=\lim_{\beta\downarrow 0}\braket{\gamma|\hat \gamma_R a|C_\beta}
\end{equation}
where in the second equality we have used equation~\eqref{eq:phi} with $|\psi\rangle=|\gamma\rangle$ and $|x\rangle =|C_\beta\rangle$, and where the third equality used~\eqref{eq:com2} (again with $|\psi\rangle=|\gamma\rangle$).

We now insert another identity operator and perform similar manipulations to write
\begin{equation}
\label{eq:5_3}
    \hat \gamma_R^\dagger \ket \gamma=\lim_{\epsilon\downarrow 0} \hat \gamma_R^\dagger \Phi_{LR}(\widehat{C_\epsilon}_L)\ket \gamma=\lim_{\epsilon\downarrow 0}\hat \gamma_R^\dagger\hat \gamma_R \ket {C_\epsilon}=\lim_{\epsilon\downarrow 0} |\hat \gamma_R|^2 \ket{C_\epsilon}.
\end{equation}
Since $tr(|\hat \gamma_R|^2 )=tr(\hat \gamma_R^\dagger\hat \gamma_R)=\braket{\gamma|\gamma}$ is finite, 
 Corollary 1 of~\cite{Colafranceschi:2023urj} then shows that the limit of $|\hat \gamma_R | |C_\epsilon \rangle$ converges as $\epsilon \downarrow 0$ to a state $|\gamma_{LL}\rangle \in \mathcal H_{LL}$.  Thus \eqref{eq:5_3} yields
 \begin{equation}
     \hat \gamma_R^\dagger \ket \gamma= |\hat \gamma_R| \ |\gamma_{LL}\rangle
 \end{equation}
 and we may evaluate \eqref{eq:weak} by writing
 \begin{equation}
\label{eq:5_4}
    \lim_{\beta\downarrow 0}\braket{\gamma|\hat \gamma_R a |C_\beta}=\lim_{\beta\downarrow 0}\braket{\gamma_{LL}|\ |\hat \gamma_R| a\ |C_\beta}=\lim_{\beta\downarrow 0}\braket{\gamma_{LL}|\ a |\hat \gamma_R|\ |C_\beta}=\braket{\gamma_{LL}| a|\gamma_{LL}},
\end{equation}
where the second step follows by noting that $|\hat \gamma_R|$ lies in the right von Neumann algebra $\mathcal A_R^{LL}$ and so necessarily commutes with $a \in \mathcal A_L^{LL}$.
Combining \eqref{eq:5_4} with \eqref{eq:weak} then completes the desired proof.
\end{proof}

We have now acquired all of the tools we need to prove continuity of the map $\Phi_{LR}$ in the weak operator topology. Recall that convergence in the weak operator topology  of the net of operators $\{\mathcal O_\alpha\}$ on a Hilbert space $\mathcal H$ is equivalent to convergence of the associated nets of matrix elements $\{\langle \phi_1|\mathcal O_{\alpha}|\phi_2\rangle\}$ for all $|\phi_1\rangle, |\phi_2\rangle \in \mathcal H$. In particular, let us consider any state $|\gamma\rangle \in \mathcal H_{LR}$ and the associated state $|\gamma_{LL}\rangle \in \mathcal H_{LL}$ defined as in Claim~\ref{thm:diag}. 
If a net of operators $\{  a_{\alpha}\}\subset \mathcal A^{LL}_L$ converges in the weak operator topology to $b \in \mathcal A^{LL}_L$, then the net of expectation values 
$\{\braket{\gamma_{LL}|a_{\alpha}|\gamma_{LL}}\}$ clearly converges to $\braket{\gamma_{LL}|b|\gamma_{LL}}$.  Claim~\ref{thm:diag} then implies that
the net of expectation values 
$\{\braket{\gamma| \Phi_{LR}(a_{\alpha})|\gamma}\}$ also converges to $\{\braket{\gamma|\Phi_{LR}(b)|\gamma} \}$.

Using the standard construction of general matrix elements from expectation values, this is enough to establish that the net of matrix elements $\{\langle \phi_1|\Phi_{LR}(a_{\alpha})|\phi_2\rangle\}$ also converges to $\langle \phi_1|\Phi_{LR}(b)|\phi_2\rangle.$  In particular, given any two states $\ket{\phi_1}, \ket{\phi_2}$, we can define $\ket \alpha=\ket {\phi_1}+\ket {\phi_2}$ and $\ket \beta=\ket {\phi_2}+i\ket {\phi_1}$ to write (for any operator $a$)
\begin{equation}
    \braket{\phi_1|a|\phi_2}=(\braket{\alpha| a|\alpha}+i\braket{\beta|a|\beta})/2-\braket{\phi_1|a|\phi_1}-\braket{\phi_2|a|\phi_2}.
\end{equation}
Convergence of expectation values of a net of operators thus implies convergence of all matrix elements.  In particular, we see from the above that the net $\{ 
\Phi_{LR}(a_\alpha) \}$ converges to $\Phi_{LR}(b)$ in the weak operator topology.

We thus see that $\Phi_{LR}$ is continuous in the weak operator topology.  This fact can be used to provide an alternate proof that $\Phi_{LR}$ is a $*$-homomorphism directly from the corresponding properties of the map $\hat \Phi_{LR}: \hat A_L^{LL} \rightarrow \hat A_L^{LR}$. And since the von Neumann algebras $\mathcal A_L^{LL}, \mathcal A_L^{LR}$ are the weak-operator-topology closures of $\hat A_L^{LL}, \hat A_L^{LR}$, surjectivity of $\hat \Phi_{LR}: \hat A_L^{LL} \rightarrow \hat A_L^{LR}$ implies 
surjectivity of $\Phi_{LR}:\mathcal A_L^{LL} \rightarrow \mathcal A_L^{LR}$.
As a result, the off-diagonal von Neumann algebra 
$\mathcal A_L^{LR}$ is isomorphic to the quotient $\mathcal A_L^{LL}/\ker(\Phi_{LR})$.
This establishes analogues of the results of section \ref{sec:alg} at the level of the corresponding von Neumann algebras.

Furthermore, since distinct sectors $\mathcal H_{\mathcal B}$ are orthogonal, the above observations achieve our primary goal.  In particular, they show agreement between the von Neumann algebra $\mathcal A_L^{LL} = \mathcal A_L^{B_L \sqcup B_L}$ (defined from the surface algebra $A_L^{B_L}$ by the diagonal sector $\mathcal H_{LL} = \mathcal H_{B_L \sqcup B_L}$) and the von Neumann algebra defined by using the representation of same surface algebra $A_L^{B_L}$ on the larger Hilbert space 
\begin{equation}
\label{eq:sumoverBR}
{\mathbb H}_{B_L}: =    
\oplus_{B_R} \mathcal H_{B_L \sqcup B_R}.
\end{equation} Here the sum on the right is over all possible right boundaries $B_R$ (including the empty set).  This is the largest Hilbert space defined by the Euclidean path integral on which $A_L^{B_L}$  can naturally be said to act.  In this sense the 
von Neumann algebra $\mathcal A_L^{LL} = \mathcal A_L^{B_L \sqcup B_L}$ defined by the diagonal sector $\mathcal H_{LL} = \mathcal H_{B_L \sqcup B_L}$ coincides with the von Neumann algebra acting on $B_L$ defined by the full quantum gravity Hilbert space.

\section{Off-diagonal Central Projections}
\label{sec:proj}

In order to pave the way for a discussion of entropy in section \ref{sec:TrEnt}, we devote this section to developing a better understanding of the off-diagonal central alebras and their relations to one another.
We have already seen in section \ref{sec:con} that the off-diagonal von Neumann algebras $\mathcal A_L^{LR}$ are quotients $\mathcal A_L^{LL}/\ker(\Phi_{LR})$ of the digaonal von Neumann algebra  $\mathcal A_L^{LL}$.  Section \ref{sec:centralproj} will make this more concrete by developing a better understanding of $\ker(\Phi_{LR})$.     Section \ref{sec:univ} 
then utilizes this understanding to 
show that our structure defines a universal central algebra, independent of any choice of boundaries, that can be said to contain all centers $\mathcal Z_{B_1\sqcup B_2} \subset \mathcal A_L^{B_1\sqcup B_2} \cap \mathcal A_R^{B_2 \sqcup B_2}$.  This will in turn help to organize our discussion of entropy in section \ref{sec:TrEnt}.

\subsection{Off-diagonal algebras as central projection of diagonal algebras}
\label{sec:centralproj}

The weak-operator-topology continuity of $\Phi_{LR}$ established in section \ref{sec:con} means that the inverse image of any weak-operator-topology-closed set is weak-operator-topology closed.  This is in particular true of $\ker(\Phi_{LR}) = \Phi_{LR}^{-1}(0)$, since any single point is weak-operator-topology closed.  Furthermore, as usual, the fact that $\Phi_{LR}$ is a homomorphsism implies that $\ker(\Phi_{LR})$ is a two-sided ideal.  We may thus make use of theorem 6.8.8 of e.g.  \cite{KR1}, which states that any weak-operator-topology-closed two-sided ideal in a von Neumann algebra $\mathcal A$ is of the form $\mathcal P \mathcal A =\mathcal A  \mathcal P$ for some central projection $\mathcal P \in \mathcal A$.  We denote the projection corresponding to $\ker(\Phi_{LR}) \subset \mathcal A_L^{LL}$ by $P_{\ker(\Phi_{LR})}$, and we also define $P^\perp_{\ker(\Phi_{LR})} = {\mathds 1} - P_{\ker(\Phi_{LR})}$, so that we have
\begin{eqnarray}
\label{eq:1minPA}
\ker(\Phi_{LR}) &=& P_{\ker(\Phi_{LR})} \mathcal A_L^{LL} \ \  \ {\rm and} \cr
    \mathcal A^{LL}_L/\ker(\Phi) & \simeq& P_{\ker(\Phi_{LR})}^\perp \mathcal A_L^{LL}
    =   \mathcal A_L^{LL} P_{\ker(\Phi_{LR})}^\perp= \mathcal A_L^{LR},
\end{eqnarray}
where the final step used the fact that $\Phi_{LR}$ defines an isomorphism between $P_{\ker(\Phi_{LR})}^\perp \mathcal A_L^{LL}$ and $\mathcal A_L^{LR}$ in order to identify these algebras.

Now, the diagonal von Neumann algebra $\mathcal A_L^{LL}$ admits a central decomposition into a direct sum which, 
due to the trace inequality \eqref{eq:trIn}, is in fact a sum over a discrete set  of type I factors indexed by $\mathcal I_{B_L}$ 
(see section 4.2 of \cite{Colafranceschi:2023urj}).  Each factor $\mathcal A^{LL}_{L,\mu}$ (for $\mu \in \mathcal I_{B_L}$) is associated with a central projection $P_\mu$ that is minimal  within the set of central projections (in the sense that there is no smaller non-trivial central projection) and which is orthogonal to $p_\nu$ when $\nu \neq \mu$. In particular, we have 
\begin{equation}
\label{eq:ALLLdecomp}
    \mathcal A_L^{LL}=\bigoplus_{\mu \in \mathcal I_{B_L}}\mathcal A^{LL}_{L,\mu} \ \ \ {\rm for} \ \ \ \mathcal A^{LL}_{L,\mu}=P_\mu \mathcal A_L^{LL}.
\end{equation}

Note that, since $P_\mu$ is a minimal central projection, we must have either $P_{\ker(\Phi_{LR})}^\perp  P_\mu=0$ or $P_{\ker(\Phi_{LR})} P_\mu=0$ (else one of these would be a smaller central projection). The former case requires $\Phi_{LR}(P_\mu)=0$, while in the latter case  we have $P_{\ker(\Phi_{LR})}^\perp P_\mu=P_\mu \neq 0$ so that 
$\Phi_{LR}(P_\mu)$ is a non-trivial minimal central projection in $\mathcal A^{LR}_L$.
We may thus write 
\begin{equation}
\label{eq:1minP}
    P_{\ker(\Phi_{LR})}^\perp =\bigoplus_{\mu\in \mathcal I_{B_L}} \chi^{LR}_{L,\mu} P_\mu,\quad 
\end{equation}
where $\chi^{LR}_{L,\mu} = 0$ when $\Phi_{LR}(P_\mu)=0$ and $\chi^{LR}_{L,\mu} = 1$ when 
$\Phi_{LR}(P_\mu)\neq 0$.  Furthermore, \eqref{eq:1minP} immediately leads to a corresponding decomposition of the algebras \eqref{eq:1minPA}.

As reviewed in section \ref{sec:rev}, the diagonal Hilbert space decomposes as a corresponding direct sum 
\begin{equation}
    \mathcal H_{B_L\sqcup B_L}=\bigoplus_{\mu \in \mathcal I_{B_L}} \mathcal H^{LL}_\mu
\ \ \ {\rm with} \ \ \ \mathcal H_\mu^{LL}=P_\mu \mathcal H_{LL}.
\end{equation} 
Furthermore, each $\mathcal H^{LL}_\mu$ admits a factorization 
\begin{equation}
    \mathcal H^{LL}_\mu=\mathcal H_{L,\mu}^{LL}\otimes \mathcal H^{LL}_{R,\mu},
\end{equation}
such that $\mathcal A^{LL}_{L,\mu}$ acts as $\mathcal B(\mathcal H_{L,\mu}^{LL})\otimes \mathds 1_\mu^R$, where $\mathds 1_\mu^R$ is the identity on the right factor $\mathcal H^{LL}_{R,\mu}$.
Since applying $\Phi_{LR}$ to \eqref{eq:ALLLdecomp} is equivalent to multiplying by  $P_{\ker(\Phi_{LR})}^\perp$, the relation \eqref{eq:1minP} yields a corresponding decomposition of the off-diagonal algebra
\begin{equation}
\label{eq:LRasum}
    \mathcal A_L^{LR} =\bigoplus_{\mu\in \mathcal I_{B_L}} \chi^{LR}_{L,\mu} \Phi_{LR}(\mathcal A^{LL}_{L,\mu})\ \ \ {\rm with}  \ \ \ \Phi_{LR}(\mathcal A^{LL}_{L,\mu})=\Phi_{LR}(P_\mu) \mathcal A^{LR}_L,
\end{equation}
and also for the off-diagonal Hilbert space
\begin{equation}
\label{eq:sum}
    \mathcal H_{LR}=\bigoplus_{\mu\in\mathcal I_{B_L} }\chi^{LR}_{L,\mu} \Phi_{LR}(P_\mu)\mathcal H_{LR}.
\end{equation}
In both \eqref{eq:LRasum} and \eqref{eq:sum} the factor of $\chi^{LR}_{L,\mu}$ is redundant (since e.g. $\chi^{LR}_{L,\mu} \Phi_{LR}(\mathcal A^{LL}_{L,\mu})=\Phi_{LR}(\mathcal A^{LL}_{L,\mu})$) but serves to emphasize which terms are non-zero.

We also emphasize that $\Phi_{LR}$ is an {\it isomorphism} when acting on any factor $\mathcal A^{LL}_{L,\mu}$ that it does not annihilate.  Thus each $\mathcal A^{LR}_{L,\mu} : = \Phi_{LR}(\mathcal A^{LL}_{L,\mu})$ is a type-I factor and the Hilbert space $\mathcal H^{LR}_\mu := \Phi_{LR}(P_\mu)\mathcal H_{LR}$ must again factorize as $\mathcal H^{LR}_{\mu}= {\mathcal  H}^{LR}_{L,\mu}\otimes \mathcal H^{LR}_{R, \mu}$ for some ${\mathcal  H}^{LR}_{L,\mu}$, $\mathcal H^{LR}_{R, \mu}$ such that 
$\mathcal A^{LR}_{L,\mu} : =\Phi_{LR}(\mathcal A^{LL}_{L,\mu})$ acts as $\mathcal B({\mathcal H}_{L,\mu}^{LR})\otimes \mathds 1_{\mu}^{R}$.  Up to the choice of an arbitrary overall phase, the isomorphism between $\mathcal A^{LL}_{L,\mu}$ and $\mathcal A^{LR}_{L,\mu}$ then also defines an isomorphism betwen the left Hilbert-space factors $\mathcal H^{LL}_{L,\mu}$ and ${\mathcal H}^{LR}_{L, \mu}$.  We will thus henceforth write $\mathcal H^{LL}_{L,\mu} = {\mathcal H}^{LR}_{L, \mu}$

The algebra $\mathcal B({\mathcal H}_{R,\mu}^{LR})$ of bounded operators on the right factor can be similarly associated with a type-I von Neumann factor given by a central projection of the right diagonal von Neumann algebra $\mathcal A^{RR}_R =\mathcal A^{B_R \sqcup B_R}_R$.  To do so, we recall that a corresponding result for diagonal Hilbert space sectors $\mathcal H_{B\sqcup B}$ was derived in~\cite{Colafranceschi:2023urj} by using the commutation theorem~\cite{Rieffel:1976} for Hilbert semi-birigged spaces derived in to show that the left and right von Neumann algebras are commutants\footnote{Though there was also a more direct proof in the diagonal case studied in \cite{Colafranceschi:2023urj}.}.  One may then check that the off-diagonal dense subspace $X=\mathcal D_{LR}$ equipped with  representations $C=\hat A^{LR}_L$ and $D=\hat A^{LR}_R$  also satisfies the four axioms of Hilbert-semi-birigged spaces from \cite{Rieffel:1976}, and where the notation $X,C,D$ comes from that reference.  One may also check that the so-called coupling condition from \cite{Rieffel:1976} is satisfied.  See appendix \ref{app:Riefel} for further details. 

As a result, theorem 1.3 of \cite{Rieffel:1976} implies that the von Neumann algebras $\mathcal A^{LR}_L$ and $\mathcal A^{LR}_R$ generated by $A^{LR}_L$ and $A^{LR}_R$ are commutants on $\mathcal H_{LR}$.  This fact has two immediate implications.  The first is that the central projections of $\mathcal A^{LR}_L$ are exactly the central projections of $\mathcal A^{LR}_R$.  In particular, the minimal central projections of  
$\mathcal A^{LR}_R$ must again be $\Phi_{LR}(P_\mu)$, so that the right algebra admits a decomposition of the form \eqref{eq:LRasum}:
\begin{equation}
\label{eq:LRasumR}
    \mathcal A_R^{LR} =\bigoplus_{\mu\in \mathcal I_{B_L}} \chi^{LR}_{L,\mu} \mathcal A^{LR}_{R,\mu},
\end{equation}
where for $\chi^{LR}_{L,\mu} =1$ the corresponding $A^{LR}_{R,\mu}=\Phi_{LR}(P_\mu)\mathcal A^{LR}_R$  is a type-I factor.

The second implication is that (again for  $\chi^{LR}_{L,\mu} =1$) each $A^{LR}_{R,\mu}$ must act as 
$\mathds 1_{\mu}^{L} \otimes B({\mathcal H}_{R,\mu}^{LR})$ on the Hilbert space $\mathcal H^{LR}_{\mu}= {\mathcal  H}^{LR}_{L,\mu}\otimes \mathcal H^{LR}_{R, \mu}$.  
We have thus arrived at an off-diagonal analogue of the structure derived for diagonal sectors in section 4 of \cite{Colafranceschi:2023urj}.  In particular, if we define $I_{B_L\sqcup B_R} \subset I_{B_L}$ as the set on which $\chi^{LR}_{L,\mu}=1$, we may write
\begin{equation}
\label{eq:fac}
    \mathcal H_{LR}=\bigoplus_{\mu\in\mathcal I_{B_L\sqcup B_R}} \mathcal H_{ \mu}^{LR} = \bigoplus_{\mu\in\mathcal I_{B_L\sqcup B_R}} \mathcal H^{LR}_{L,\mu}\otimes \mathcal H_{R,\mu}^{LR}.
\end{equation}
Now, as described above, we may use the isomorphism $\Phi_{LR}: \mathcal A^{LL}_{L,\mu} \rightarrow \mathcal A^{LR}_{L,\mu}$ for $\mu \in \mathcal I_{B_L\sqcup B_R}$ to write $\mathcal H_{L,\mu}^{LR} = \mathcal H_{L,\mu}^{LL}$ for such $\mu$.  But since all of the above discussion of the left von Neumann algebras $\mathcal A^{LL}_L$ can be repeated analogously for the right von Neumann algebras $\mathcal A^{RR}_R$, there must be another surjective $*$-homomorphism  from $\mathcal A_R^{RR} \rightarrow \mathcal A_R^{LR}$. We will call this new homomorphism $\Phi^R_{LR}$ and, for clarity, we will sometimes also write $\Phi^L_{LR}$ for the previous left map $\Phi_{LR}$.   Furthermore, for $\mu \in \mathcal I_{B_L \sqcup B_R}$, the right homomorphism $\Phi^R_{LR}$ must define an isomorphism between $\mathcal A_{R, \mu}^{LR}$ and some type-I von Neumann factor $\mathcal A_{R, \mu}^{RR}$ in  $\mathcal A_{R}^{RR}$.  We may then also write $\mathcal H^{LR}_{R,\mu} = \mathcal H^{RR}_{R,\mu}$, and thus
\begin{equation}
\label{eq:fac2}
    \mathcal H_{LR}= \bigoplus_{\mu\in\mathcal I_{B_L\sqcup B_R}} \mathcal H^{LR}_{L,\mu}\otimes \mathcal H_{R,\mu}^{LR} = \bigoplus_{\mu\in\mathcal I_{B_L\sqcup B_R}} \mathcal H^{LL}_{L,\mu}\otimes \mathcal H_{R,\mu}^{RR};
\end{equation}
i.e., the decomposition of the off-diagonal sectors involves the same left and right Hilbert spaces as the decomposition of the digaonal factors.

\subsection{A universal central algebra for $B$}
\label{sec:univ}

Our result \eqref{eq:fac2} imposes a compatibility constraint on our von Neumann algebras. Since the left algebra $\mathcal A_L^{B_1\sqcup B_2}$ and the right algebra $\mathcal A_R^{B_1\sqcup B_2}$ associated with two arbitrary boundaries $B_1$ and $B_2$ must be commutants, they share a common center $\mathcal Z^{B_1\sqcup B_2}$. Furthermore, we saw that the  center of such an off-diagonal algebra is identified with  central projections of the centers of either diagonal algebra $\mathcal A_L^{B_1 \sqcup B_1}$, $\mathcal A_R^{B_2 \sqcup B_2}$. We thus have the following relations:
\begin{equation}
\label{eq:compat}
    \mathcal Z^{B_1\sqcup B_2}\simeq  P^{B_1\sqcup B_2}_L \mathcal Z^{B_1\sqcup B_1}\simeq  P^{B_1\sqcup B_2}_R \mathcal Z^{B_2\sqcup B_2}.
\end{equation}

Let us now define a universal Abelian algebra
\begin{equation}
\label{eq:center}
\Tilde {\mathcal Z}_{univ}  =\bigoplus_B \mathcal Z^{B\sqcup B}
\end{equation}
 as the direct sum of all diagonal centers.  We may again regard this as a von Neumann algebra by taking the elements to act on the direct sum Hilbert space
 $(\bigoplus_B \mathcal H_{B\sqcup B})$.
 
As it stands, the algebra \eqref{eq:center} is a purely formal construction whose universality comes by fiat from our choice to sum over all source-manifolds $B$.  However, the result~\eqref{eq:compat} turns out to define an interesting equivalence relation $\sim$ on elements of $\Tilde{\mathcal Z}^{univ}$ so that the quotient ${\mathcal Z}^{univ}=\Tilde{\mathcal Z}^{univ}/\sim$ provides an algebraic encoding of the above compatibility condition. To see this, we begin by defining a relation $\sim$ that relates centrally-minimal projections in different centers:

\begin{definition}
Consider projections $P_1 \in \mathcal Z^{B_1\sqcup B_1}$, $P_2 \in \mathcal Z^{B_2\sqcup B_2}$ that are each minimal in their respective central algebras.  We will write $P_1 \sim P_2$ when there is some boundary $B$,  a projection $P\in \mathcal Z^{B\sqcup B}$, and non-zero states $|\phi\rangle \in {\cal H}_{B_1 \sqcup B}$, $|\psi\rangle \in {\cal H}_{B_2 \sqcup B}$  that satisfy
\begin{eqnarray}
\label{eq:phipsi}
\Phi^L_{B_1 B} (P_1) |\phi \rangle &=& \Phi^R_{B_1 B} (P)|\phi \rangle  \neq 0 \ \ \ {\rm and} \cr
\Phi^L_{B_2 B} (P_2) |\psi \rangle &=& \Phi^R_{B_2 B} (P)|\psi \rangle  \neq 0.
\end{eqnarray}
\end{definition}
Here $\Phi^L_{BB_2}, \Phi^L_{BB_1}$ are just $\Phi^L_{LR}$ for $B_L=B$ and $B_R= B_2, B_1$, and we will use analogous notation below. 

This definition is manifestly symmetric, meaning that $P_1 \sim P_2$ is equivalent to $P_2 \sim P_1$.  It also satisfies the reflexive property $P_1 \sim P_1$.  This may be seen by setting $B=B_2 = B_1$, $P = P_2 =P_1$, and $|\psi\rangle = |\phi\rangle = P_1 |C_\beta\rangle \in \mathcal H_{B_1 \sqcup B_1}$ 
for some $\beta >0$.  Since it was shown in \cite{Colafranceschi:2023urj} that $P_1 |C_\beta\rangle$ cannot vanish for any central projection $P_1 \in \mathcal Z^{B_1 \sqcup B_1}$, to establish $P_1 \sim P_1$ in this context we need only recall that any such $P_1$ can be interpreted as a member of both $\mathcal A_L^{B_1 \sqcup B_1}$ and $\mathcal A_R^{B_1 \sqcup B_1}$ and that $\Phi^L_{B_1 B_1}(P_1) = P_1 = \Phi^R_{B_1 B_1}(P_1)$.

Showing that $\sim$ is an equivalence relation thus requires only that we establish transitivity, which means that $P_2 \sim P_3$ whenever $P_1 \sim P_2$ and $P_1 \sim P_3$.  The conditions   $P_1 \sim P_2$ and $P_1 \sim P_3$ mean that there are boundaries $B,B'$, projections $P\in \mathcal Z^{B\sqcup B}$, $P'\in \mathcal Z^{B'\sqcup B'}$, and non-zero states $\phi_{B_1B}, \phi_{B_2B}, \phi_{B_1B'}, \phi_{B_3 B'}$ (which respectively lie in ${\cal H}_{B_1 \sqcup B}$, ${\cal H}_{B_2 \sqcup B}$, ${\cal H}_{B_1 \sqcup B'}$, ${\cal H}_{B_3 \sqcup B'}$) such that we have 
\begin{eqnarray}
\label{eq:phi1B2B}
\Phi^L_{B_1 B} (P_1) |\phi_{B_1B} \rangle &=& \Phi^R_{B_1 B} (P)|\phi_{B_1B} \rangle  \neq 0, \ \ \  \cr
\Phi^L_{B_2 B} (P_2) |\phi_{B_2B} \rangle &=& \Phi^R_{B_2 B} (P)|\phi_{B_2B} \rangle  \neq 0,
\cr
\Phi^L_{B_1 B'} (P_1) |\phi_{B_1B'} \rangle &=& \Phi^R_{B_1 B'} (P')|\phi_{B_1B'} \rangle  \neq 0, \ \ \ \rm{and}
\cr
\Phi^L_{B_3 B'} (P_3) |\phi_{B_3 B'} \rangle &=& \Phi^R_{B_3 B'} (P')|\phi_{B_3 B'} \rangle  \neq 0.
\end{eqnarray}

Given the existence of $|\phi_{B_3B'} \rangle$, we can show $P_1\sim P_3$ by demonstrating the existence of a non-zero $|\phi_{B_2 B'} \rangle \in {\cal H}_{B_2 \sqcup B'}$ such that
\begin{equation}
\Phi^L_{B_2 B'} (P_2) |\phi_{B_2 B'} \rangle = 
\Phi^R_{B_2 B'} (P')|\phi_{B_2 B'} \rangle  \neq 0.
\end{equation}

We will build such a $|\phi_{B_2 B'} \rangle$ by noting that
the construction of our map $\Psi_R: \mathcal H_{LR} \rightarrow \mathcal B(\mathcal H_{LL}, \mathcal H_{LR})$ readily generalizes to define a map  $\Psi_R^{B, B_1 \rightarrow B_2}: \mathcal H_{B_1 \sqcup B_2} \rightarrow \mathcal B(\mathcal H_{B \sqcup B_1}, \mathcal H_{B \sqcup B_2})$ involving arbitrary Hilbert spaces $B, B_1, B_2$.  We will again use the simplified notation $\hat a_R : = 
\Psi_R^{B, B_1 \rightarrow B_2}(|a\rangle).$
As with the original $\Psi_R$, the full $\Psi_R^{B, B_1 \rightarrow B_2}$ is first defined as a map
from 
$\mathcal D_{B_1 \sqcup B_2}$ to $\mathcal B(\mathcal D_{B \sqcup B_1}, \mathcal D_{B \sqcup B_2})$ where it is defined by sewing surfaces as in the definition of $\hat a_R$ in figure \ref{fig:glue}.  As described in appendix~\ref{app:Psi}, in the same manner that it was demonstrated for the original $\Psi_R$, each $\hat a_R$ defined by
the more general $\Psi_R^{B, B_1 \rightarrow B_2}$
is a bounded operator and, moreover, the map  $\Psi_R^{B, B_1 \rightarrow B_2}$ itself
can be shown to be continuous with respect to the Hilbert space topology on $\mathcal H_{B_1 \sqcup B_2}$  and the strong operator topology on $B(\mathcal H_{B \sqcup B_1}, \mathcal H_{B \sqcup B_2})$. The full map $\Psi_R^{B, B_1 \rightarrow B_2}$ is then given by the unique continuous extension to the full Hilbert space $\mathcal H_{B_1 \sqcup B_2}$.   Appendix~\ref{app:Psi} also derives two intertwining relations \eqref{eq:Bint}, \eqref{eq:fullPsiadj}, which we restate here for the convenience of the reader.  The first of these is that 
for all $d \in \mathcal A_L^{B \sqcup B}$ we have 
\begin{equation}
\label{eq:Bint2}
\Phi_{B B_2}^L(d) \hat a_R =  \hat a_R  \Phi^L_{B B_1}(d). 
\end{equation}
 The second is the relation
\begin{equation}
\label{eq:fullPsiadj2}
    {\left(a^\star\right)}_R: =\hat a_R^\dagger,
\end{equation}
where (as described in appendix~\ref{app:Psi}) the operation $\star$ has been extended from $\mathcal D_{B_1 \sqcup B_2}$ to the entire Hilbert space by continuity.  

There is also an analogous left map 
$\Psi_L^{B_2 \rightarrow B_1, B}: \mathcal H_{B_1 \sqcup B_2} \rightarrow \mathcal B(\mathcal H_{B_2 \sqcup B}, \mathcal H_{B_1 \sqcup B})$. Writing, 
$\hat a_L : = 
\Psi_L^{B_2 \rightarrow B_1,B}(|a\rangle)$, this map satisfies 
\begin{equation}
\label{eq:Bint2L}
\Phi^R_{B_1 B}(d) \hat a_L =  \hat a_L  \Phi^R_{B_2 B}(d) 
\end{equation}
 for all $d \in \mathcal A_R^{B \sqcup B}$. It also satisfies
\begin{equation}
\label{eq:fullPsiadj2L}
    {\left(a^\star\right)}_L: =\hat a_L^\dagger.
\end{equation}

By including the projections $\Phi^L_{B_2 B}(P_2), \Phi^R_{B_1 B'}(P')$ in their definitions,  we can take $|\phi_{B_2 B}\rangle$ and $|\phi_{B_1 B'}\rangle$ to satisfy
\begin{equation}
\label{eq:projconds}
\Phi^L_{B_2 B}(P_2) |\phi_{B_2 B}\rangle = |\phi_{B_2 B}\rangle \neq 0, \ \ \ 
\Phi^R_{B_1 B'}(P') |\phi_{B_1 B'}\rangle = |\phi_{B_1 B'}\rangle \neq 0.
\end{equation}
If we now choose any
$b_R \in \mathcal A^R_{B_1 B_1}, c_R \in \mathcal A^R_{BB}$,
the above structures allow us to define a state 
\begin{equation}
\label{eq:phiB2Bprime}
    |\phi_{B_2 B'} \rangle =  \widehat{\left(\phi_{B_1 B'}\right)}_R \ \Phi^R_{B_2B_1}(b_R) \ \widehat{\left(\phi_{B_1 B}\right)}_R^\dagger \Phi^R_{B_2 B}(c_R) \ |\phi_{B_2 B} \rangle,
\end{equation}
If all of the relevant states and operators above were defined directly by rimmed surfaces we could write
\begin{eqnarray}
\label{eq:phiB2Bprime1}
    |\phi_{B_2 B'} \rangle &=&  \widehat{\left(\phi_{B_1 B'}\right)}_R \ \Phi^R_{B_2B_1}(\hat b_R) \ \widehat{\left(\phi_{B_1 B}\right)}_R^\dagger \Phi^R_{B_2 B}(
    \hat c_R) \ |\phi_{B_2 B} \rangle \cr
= |  \phi_{B_2 B}\ c \ \phi_{B_1 B}^\star \ b \  \phi_{B_1 B'} \rangle
&=&  \widehat{\left(\phi_{B_2 B}\right)}_L \ \Phi^L_{BB'}(\hat c_L) \ \widehat{\left(\phi_{B_1 B}\right)}_L^\dagger \Phi^L_{B_1 B'}(\hat b_L) \  |\phi_{B_1 B'} \rangle,
\end{eqnarray}
As described in appendix \ref{app:Ldef}, this observation generalizes to arbitrary states and operators as above when written in the form
\begin{eqnarray}
\label{eq:phiB2Bprime2}
    |\phi_{B_2 B'} \rangle &=&  \widehat{\left(\phi_{B_1 B'}\right)}_R \ \Phi^R_{B_2B_1}(b_R) \ \widehat{\left(\phi_{B_1 B}\right)}_R^\dagger \Phi^R_{B_2 B}(c_R) \ |\phi_{B_2 B} \rangle \cr
&=&  \widehat{\left(\phi_{B_2 B}\right)}_L \ \Phi^L_{BB'}(c_L) \ \widehat{\left(\phi_{B_1 B}\right)}_L^\dagger \Phi^L_{B_1 B'}(b_L) \  |\phi_{B_1 B'} \rangle,
\end{eqnarray}
where $c_L,b_L$ are defined from $b_R, c_R$ as described in appendix \ref{app:Ldef}.

We may now compute
\begin{eqnarray}
    \Phi^L_{B_2B'}(P_2)    |\phi_{B_2 B'} \rangle &=&  \Phi^L_{B_2B'}(P_2) \ \widehat{\phi_{B_1 B'}}_R \ \Phi^R_{B_2B'}(b_R) \widehat{\phi_{B_1 B}}_R^\dagger\  \Phi^R_{B_2 B}(c_R)
  \ |\phi_{B_2 B} \rangle \cr &=& \widehat{\phi_{B_1 B'}}_R \ \Phi^R_{B_2B'}(b_R) \ \widehat{\phi_{B_1 B}}_R^\dagger \ \Phi^R_{B_2 B}(c_R)  \
 \Phi^L_{B_2B}(P_2)  |\phi_{B_2 B} \rangle \cr
 &=&  |\phi_{B_2 B'} \rangle,
\end{eqnarray}
where we pass from the first to the 2nd line using intertwining relations of the form \eqref{eq:Bint2} and commutivity of left- and right-acting operators, and where the final step uses \eqref{eq:projconds}. 
We may also similarly write 
\begin{eqnarray}
    \Phi^R_{B_2B'}(P')    |\phi_{B_2 B'} \rangle &=&  \Phi^R_{B_2B'}(P')    
    \widehat{\left(\phi_{B_2 B}\right)}_L \ \Phi^L_{BB'}(c_L) \ \widehat{\left(\phi_{B_1 B}\right)}_L^\dagger \Phi^L_{B_1 B'}(b_L) \    |\phi_{B_1 B'} \rangle \cr  &=& 
    \widehat{\left(\phi_{B_2 B}\right)}_L \ \Phi^L_{BB'}(c_L) \ \widehat{\left(\phi_{B_1 B}\right)}_L^\dagger \Phi^L_{B_1 B'}(b_L) \    \Phi^R_{B_1B'}(P')  |\phi_{B_1 B'} \rangle\cr
 &=&  |\phi_{B_2 B'} \rangle.
\end{eqnarray}

This will establish transitivity so long as we also show that we can choose $b_R,c_R$ such that $|\phi_{B_2 B'} \rangle \neq 0$.  We can do so by using the following two results (where the first will act as a Lemma that will be useful in proving the second.

\begin{theorem}
\label{thm:nonzero}
For non-zero $|\phi\rangle, |\kappa\rangle \in {\cal H}_{B_1 \sqcup B_2}$, the states
$\hat \phi_R^\dagger| \phi\rangle \in {\cal H}_{B_1 \sqcup B_1}$ and
$\hat \kappa_R |\kappa^\star\rangle \in {\cal H}_{B_2 \sqcup B_2}$ cannot vanish. Here $\hat \phi_R : = \Psi_R^{B_1, B_1 \rightarrow B_2}(|\phi\rangle)$ and
$\hat \kappa_R : = \Psi_R^{B_2, B_1 \rightarrow B_2}(|\kappa\rangle)$.
\end{theorem}
\begin{proof}
Let us first use \eqref{eq:Crep`} and boundedness of
$\hat \phi_R$  (and thus of its adjoint) to
rewrite the first  state in the form
\begin{equation}
\label{eq:psiC}
\hat \phi_R^\dagger|\phi\rangle
= \lim_{\beta \downarrow 0} \hat \phi_R^\dagger \hat \phi_R  |C_\beta\rangle.
\end{equation}
The definition  \eqref{eq:trvN} of the trace $tr$ then shows that  the norm of our state is $tr\left((\hat\phi_R^\dagger \hat\phi_R)^2\right)$.  But $\hat\phi_R^\dagger \hat\phi_R$ is non-zero since deleting
$\hat \phi_R^\dagger$ from
\eqref{eq:psiC} and repeating the same argument gives
$tr(\hat\phi_R^\dagger \hat\phi_R)= \langle \phi |\phi\rangle \neq 0$. Since
$\hat\phi_R^\dagger \hat\phi_R$ is manifestly self-adjoint, the operator $(\hat\phi_R^\dagger \hat\phi_R)^2$ is again non-zero and faithfulness of $tr$ as established in \cite{Colafranceschi:2023urj} means that $tr[(\hat\phi_R^\dagger \hat\phi_R)^2] \neq 0.$  Thus our state has non-zero norm and cannot vanish. The analogous argument then also shows that $\hat \kappa_R |\kappa^\star\rangle \in {\cal H}_{B_2 \sqcup B_2}$ is non-zero. 
\end{proof}

\begin{theorem}
\label{thm:trans1}
Given a centrally-minimal projection $P_1 \in \mathcal A_L^{B_1 \sqcup B_1}$ and non-zero states $|\phi\rangle \in {\cal H}_{B_1 \sqcup B_2}$ and $|\kappa\rangle \in {\cal H}_{B_1 \sqcup B_3}$ that satisfy
\begin{equation}
\label{eq:cpcs}
\Phi^L_{B_1 B_2}(P_1) |\phi\rangle =  |\phi\rangle\quad
\Phi^L_{B_1 B_3}(P_1) |\kappa\rangle =  |\kappa\rangle, 
\end{equation}
there is some $a_R$ in the right von Neumann algebra $\mathcal A_R^{B_2 \sqcup B_1}$  for which the state 
$|\gamma\rangle: = \hat \kappa_R a_R |\phi^\star \rangle \in  
{\cal H}_{B_2 \sqcup B_3}$ is non-zero, where as usual $\hat \kappa_R = \Psi_R^{B_2, B_1\to B_3}(|\kappa\rangle)$
\end{theorem}
\begin{proof}
We will also write  $\widehat {\left(\phi^\star\right)}_R := \Psi_R^{B_1,B_2\to B_1}(|\phi^\star\rangle) = \left[\Psi_R^{B_1,B_1\to B_2}(|\phi\rangle)\right]^\dagger =: \hat \phi_R^\dagger$ , where the 2nd equality uses \eqref{eq:fullPsiadj}.  
As in the proof of \ref{thm:nonzero} above, 
the operators $\hat \kappa_R^\dagger \hat \kappa_R$ and $\widehat {\left(\phi^\star\right)}_R \widehat {\left(\phi^\star\right)}_R^\dagger = \hat \phi_R^\dagger \hat \phi_R$ on $\mathcal A^L_{B_1 \sqcup B_1}$ are both non-zero. 

Furthermore, for $P_1^\perp = \mathds{1}-P_1$, we may write  e.g.,
\begin{equation}
\label{eq:P1kappa}
\lim_{\beta \downarrow 0} P_1^\perp  \hat \kappa_R^\dagger \hat \kappa_R  |C_\beta\rangle
=P_1^\perp \lim_{\beta \downarrow 0} \hat \kappa_R^\dagger \hat \kappa_R  |C_\beta\rangle
=P_1^\perp \hat \kappa_R^\dagger |\kappa\rangle
= \hat \kappa_R^\dagger \left[\Phi^L_{B_1B_2}\left(P_1^\perp\right)\right]|\kappa\rangle =0,
\end{equation}
where the first equality uses boundedness of $P_1^\perp$, the 2nd follows as in the proof of \ref{thm:nonzero}, the third uses the intertwining relation \eqref{eq:com2}, and the fourth follows from \eqref{eq:cpcs} since $\Phi^L_{B_1B_2}\left(P_1^\perp\right) = \mathds{1} - \Phi^L_{B_1B_2}\left(P_1\right)$. Using the definition \eqref{eq:trvN} of the trace $tr$ on positive elements of the von Neumann algebra $\mathcal A_R^{B_1 \sqcup B_1}$,  the above result requires $\left(\kappa_R^\dagger \hat \kappa_R  (\mathds{1}-P_1)  \hat \kappa_R^\dagger \hat \kappa_R \right)$ to have vanishing trace. 

Faithfulness of the trace now requires $\kappa_R^\dagger \hat \kappa_R  (\mathds{1}-P_1)  \hat \kappa_R^\dagger \hat \kappa_R =0$.  In particular, for all $|\gamma \rangle \in \mathcal H_{B_1 \sqcup B_1}$ we have 
\begin{equation}
\label{eq:prepforsqrt}
\langle \gamma | \kappa_R^\dagger \hat \kappa_R  (\mathds{1}-P_1)  \hat \kappa_R^\dagger \hat \kappa_R  | \gamma \rangle = 0.    
\end{equation}
But \eqref{eq:prepforsqrt} is the norm of the state 
$(\mathds{1}-P_1)  \hat \kappa_R^\dagger \hat \kappa_R  | \gamma \rangle$, so this state must vanish for all $\gamma$ and we have 
\begin{equation}
P_1  \hat \kappa_R^\dagger \hat \kappa_R = \hat \kappa_R^\dagger \hat \kappa_R.    
\end{equation}
  Since the same argument holds with $\kappa$ replaced with $\phi$, we must have
\begin{equation}
P_1 \hat \phi_R^\dagger \hat \phi_R
  = \hat \phi_R^\dagger \hat \phi_R, 
  \ \ \ {\rm or,\ equivalently, } \ \ \ 
P_1 \widehat {\left(\phi^\star\right)}_R \widehat {\left(\phi^\star\right)}_R^\dagger  =  \widehat {\left(\phi^\star\right)}_R \widehat {\left(\phi^\star\right)}_R^\dagger,
\end{equation}
so that both $\hat \kappa_R^\dagger \hat \kappa_R$ and $\widehat {\left(\phi^\star\right)}_R \widehat {\left(\phi^\star\right)}_R^\dagger$ in fact lie in the same $\mu$-sector $\mathcal A^{B_1 \sqcup B_1}_{R,\mu}$.

Due to \eqref{eq:cpcs} and the fact that $|\phi \rangle \neq 0$, we see that  $\mathcal A^{B_1 \sqcup B_1}_{R,\mu}$  is also isomorphic to some off-diagonal factor
$\mathcal A^{B_2 \sqcup B_1}_{R,\mu}$, and we may use $\Phi^R_{B_2 B_1}$ to denote this isomorphism. But let us also recall that $\mathcal A^{B_1 \sqcup B_1}_{R,\mu}$ is a type I von Neumann factor isomorphic to $\mathcal B(\mathcal H_{B_1 \sqcup B_1, R}^\mu)$.   Thus $\mathcal A^{B_2 \sqcup B_1}_{R,\mu}$ is also isomorphic to  $\mathcal B(\mathcal H_{B_1 \sqcup B_1, R}^\mu)$.  

Furthermore, since traces on type-I factors are unique up to multiplication by constants, up to such constants we can use the identification of $\mathcal A^{B_2 \sqcup B_1}_{R,\mu}$ with $\mathcal B(\mathcal H_{B_1 \sqcup B_1, R}^\mu)$ to evaluate traces on $\mathcal A^{B_2 \sqcup B_1}_{R,\mu}$.
We can then use the trace on $\mathcal B(\mathcal H_{B_1 \sqcup B_1, R}^\mu)$ to show that there is some $a_R \in \mathcal A^{B_2 \sqcup B_1}_{R,\mu}$ for which $tr(\widehat {\left(\phi^\star\right)}_R^\dagger  a_R^\dagger \hat \kappa_R^\dagger
\hat \kappa_R  a_R  \widehat {\left(\phi^\star\right)}_R )$ is non-zero.  
This is in particular true for the $a_R$ obtained by 
applying the isomorphism $\Phi^R_{B_2 B_1}$ to the operator $|\overline{\phi^\star}\rangle \langle \overline \kappa | \in \mathcal B(\mathcal H_{B_1 \sqcup B_1, R}^\mu)$ where $|\overline \kappa\rangle, |\overline{\phi^\star}\rangle \in H_{B_1 \sqcup B_1, R}^\mu$ are the largest-eigenvalue eigenvectors of 
$\hat \kappa_R^\dagger \hat \kappa_R$ and $\widehat {\left(\phi^\star\right)}_R \widehat {\left(\phi^\star\right)}_R^\dagger$. But by the same argument used in the proof of claim~\ref{thm:nonzero}, this trace is $\langle \gamma|\gamma\rangle$ for $|\gamma\rangle: = \hat \kappa_R a_R  |\phi^\star \rangle \in  
{\cal H}_{B_2 \sqcup B_3}$.  Thus $|\gamma \rangle$ is non-zero as claimed. 
\end{proof}

To show that the $|\phi_{B_2B'}\rangle$ defined in \eqref{eq:phiB2Bprime} is non-zero, we simply apply Claim \ref{thm:trans1} twice.  The first time uses $|\phi_1\rangle =|\phi_{B_2B}^\star\rangle$ and $|\kappa_1 \rangle = | \phi_{B_1 B}^\star\rangle $, while the second uses $|\kappa_2\rangle = | \phi_{B_1 B'}\rangle $ and 
\begin{equation}
    |\phi_2\rangle =  \widehat{\phi_{B_2 B}}_R^\dagger \hat c^\dagger_R |\phi_{B_1B}\rangle.
\end{equation}
The fact that $|\phi_2 \rangle \neq 0$ follows from the first application of Claim \ref{thm:trans1}, and the second application then gives $|\phi_{B_2B'}\rangle \neq 0$

This establishes the desired transitivity of $\sim$ and shows that it defines an equivalence relation on centrally-minimal projections.  We can then extend the relation $\sim$ to  more general elements of $Z^{B_1\sqcup B_1}$, $\mathcal Z^{B_2\sqcup B_2}$. We will write $z_1\sim z_2$ when $z_1 = \sum_{i} c_{i}P_{1i}$ and 
$z_2 = \sum_{i} c_{i}P_{2i}$ for centrally-minimal projections $P_{1i}, P_{2i}$ with $P_{1i} \sim P_{2i}.$
It is then manifest that this extension is again an equivalence.

We may thus take the quotient of $\Tilde{\mathcal Z}_{univ}$ by the relation $\sim$ to obtain a  smaller algebra
\begin{equation}
\label{eq:center2}
    \mathcal Z_{univ}=\Tilde{\mathcal Z}_{univ}/\sim.
\end{equation}
The algebra \eqref{eq:center2} is universal in the sense that any diagonal or non-diagonal center can be obtained as a projection of $\mathcal Z_{univ}$.  This property allows us to then identify any center $\mathcal Z^{B_1 \sqcup B_2}$ with a subalgebra of
$\mathcal Z_{univ}$.  In particular, this allows us to say that $z_{12} \in \mathcal Z^{B_1 \sqcup B_2}$ and $z_{34} \in \mathcal Z^{B_3 \sqcup B_4}$ are `the same' if they map to the same element of
$\mathcal Z_{univ}$.

\section{Traces and entropy}
\label{sec:TrEnt}

The remaining properties  (regarding entropy, hidden sectors, etc.) discussed in section 4 of \cite{Colafranceschi:2023urj} now follow by essentially the same arguments given there.  However, we can use the universal structure identified in section \ref{sec:univ} to provide a unified and clean discussion for any bipartite boundary $B_1 \sqcup B_2$.  We briefly summarize such results below.

Recall that we have a $*$-homomorphism $\Phi_{BB'}$ from $\mathcal A_L^{B \sqcup B}$ to the bounded operators on ${\mathcal H}_{B \sqcup B'}$.
Given any normalized state $|\psi\rangle \in {\mathcal H}_{B \sqcup B'}$, we can thus compute the expectation value $\langle \psi |\Phi_{BB'}(a)|\psi \rangle$ for any $a \in \mathcal A_L^{B \sqcup B}$.  This is a positive linear functional on the type-I von Neumann algebra $\mathcal A_L^{B \sqcup B}$.  Furthermore, our functional is normalized in the sense that 
\begin{equation}
\langle \psi |\Phi_{BB'}({\mathds 1})|\psi \rangle= 
\langle \psi | {\mathds 1}|\psi \rangle = 1.
\end{equation}
As a result, given any trace $tr$ on the set $\left(\mathcal A_L^{B \sqcup B}\right)^+$ of positive elements in $\mathcal A_L^{B \sqcup B}$,
there is a density operator $\rho^{tr}_\psi \in \left(\mathcal A_L^{B \sqcup B}\right)^+$ with $tr(\rho^{tr}_\psi)=1$ such that
\begin{equation}
\label{eq:densop}
tr\left( \left(\rho^{tr}_\psi\right)^{1/2} a \left(\rho^{tr}_\psi\right)^{1/2}\right) = \langle \psi |\Phi_{BB'}(a)|\psi \rangle
\end{equation}
for all $a \in \left(\mathcal A_L^{B \sqcup B}\right)^+$. A brute-force argument for \eqref{eq:densop} can be found in the discussion around (4.63) in \cite{Colafranceschi:2023urj} where $\rho^{tr}_\psi$ is constructed by tracing out the right Hilbert space factors from each term in the direct sum that defines $\mathbb{H}_B$.  However,  we suspect that the mathematical literature contains a more elegant derivation.

The notation $\rho^{tr}_\psi$ explicitly indicates the dependence of this operator on the choice of the trace $tr$.  Three traces on $\mathcal A_L^{B \sqcup B}$ were discussed in \cite{Colafranceschi:2023urj}.  The first,  called simply $tr$, is given by \eqref{eq:trvN}.  
The second trace, called $Tr$, was defined by choosing an arbitrary orthonormal basis $\{|i,\mu\rangle\}$ for each of
the Hilbert spaces $\mathcal H^{B \sqcup B}_{\mu,L}$  for $\mu \in \mathcal I_{B \sqcup B}$
and writing
\begin{equation}
Tr(a) := \sum_{i, \mu} \langle i,\mu|a|i,\mu \rangle
\end{equation}
for every $a \in \left(\mathcal A_L^{B \sqcup B}\right)^+$.

Now, on a given type-I factor $\mathcal A_{L,\mu}^{B \sqcup B}$, any two traces are proportional.  In this context we must thus have 
\begin{equation}
\label{eq:trTr}
    tr= n_\mu \, Tr.
\end{equation} It was shown in \cite{Colafranceschi:2023urj} that their axioms imply higher generalizations of the trace inequality
\eqref{eq:trIn} which require the $n_\mu$ to be positive integers.  The third trace $\widetilde {Tr}$ from \cite{Colafranceschi:2023urj} can be defined by constructing the Hilbert spaces
\begin{equation}
\label{eq:addHS}
    \undertilde {\mathcal  H}_B:= \bigoplus_\mu {\mathcal H}^{B \sqcup B}_{\mu,L} \otimes {\mathbb C}^{n_\mu},
\end{equation}
choosing an arbitrary orthonormal basis $\{\ket {\tilde i}\}$ of $\undertilde {\mathcal  H}^{B}$, and then writing
\begin{equation}
\label{eq:TildeTr}
\widetilde{Tr}(a) := \sum_{\tilde i} \langle \tilde i|\left(a\otimes \mathds 1\right)|\tilde i\rangle
\end{equation}
for every $a \in \left(\mathcal A_L^{B \sqcup B}\right)^+$.    Comparing \eqref{eq:trTr} with \eqref{eq:TildeTr} then shows that $tr=\widetilde{Tr}$, and thus that $\rho_{\psi}^{tr}=\rho_{\psi}^{\widetilde {Tr}}$.  Following \cite{Colafranceschi:2023urj}, to simplify the notation we will henceforth denote this density operator by  $\tilde \rho_{\psi}$.  The factors $ {\mathbb C}^{n_\mu}$ in \eqref{eq:addHS} were termed {\it hidden sectors} in \cite{Colafranceschi:2023urj}, where their physical interpretation was discussed in more detail and illustrated with examples\footnote{The structure of \eqref{eq:addHS} and the embedding \eqref{eq:tildeembed} of $\mathcal H_{B \sqcup B'}$ in $\undertilde{\mathcal H}_B \otimes \undertilde{\mathcal H}_{B'}$ define a (sum of) quantum error correcting codes with two-sided recovery of the type described in \cite{Harlow:2016vwg}.  This is not a coincidence, as the insertion of Hidden sectors can be taken to define the Hilbert space of a more fundamental theory in which $\mathcal H_{B \sqcup B'}$ is embedded.   Some differences, however, are that the QEC structure described in \cite{Harlow:2016vwg} is expected to hold only approximately and to require some notion of cut-off in the bulk.  In contrast, our $\mathcal H_{B \sqcup B'}$ is an exact construction.  It is also often noted that the same mathematical structure also appears in discussions of so-called edge modes (see e.g. \cite{Donnelly:2011hn,Donnelly:2014fua,Donnelly:2016auv}).  However, the physics of edge modes appears to be rather different.  In particular, in contrast to edge modes, our ${\mathbb C}^{n_\mu}$ hidden sectors play no role in the factorization of the Hilbert space.  Furthermore, there are interesting examples in which the hidden sectors are trivial.}.

An explicit formula for 
$\tilde \rho_{\psi}$ can be given by using an appropriate embedding of $\mathcal H_{B \sqcup B'}$ in $\undertilde{\mathcal H}_B \otimes \undertilde{\mathcal H}_{B'}$.    To do so, simply choose a maximally entangled state $|\chi_\mu\rangle$ in $ {\mathbb C}^{n_\mu} \otimes  {\mathbb C}^{n_\mu}$ for each $\mu$ and write
\begin{equation}
\label{eq:tildeembed}
|\tilde \psi\rangle = \sum_{P \in \mathcal Z_{univ}^{MP}} \left( \Phi_{BB'}(P)|\psi\rangle \right) \otimes |\chi_\mu\rangle,
\end{equation}
where  $\mathcal Z_{univ}^{MP}$ is the set of projections in $\mathcal Z_{univ}$ that are also minmal in $\mathcal Z_{univ}$.   In \eqref{eq:tildeembed}, we have also defined $\Phi_{BB'}(P) : =  \Phi_{BB'}(P')$ when there is some $P' \sim P$ with $P' \in \mathcal Z^{B \sqcup B}$.  When there is no such $P'$, we set $\Phi_{BB'}(P) : = 0$. Finally, we also regard $\mu$ as a function of $P$ defined by the fact that the above $P'$ here projects onto some $\mathcal H_{BB}^\mu$.

The density operator (or density matrix) $\tilde \rho_\psi$ can then be constructed from $|\tilde \psi\rangle\langle \tilde \psi |$ by `tracing out' the right factor $\undertilde {\mathcal H}_{B'}$ in the above decomposition (using its natural trace $\widetilde {Tr}$).   The result $\tilde \rho_\psi$ is a density matrix on $\undertilde {\mathcal H}_B$  whose  von Neumann entropy  is given by $S(B) : = -tr\left(\tilde \rho_\psi \ln \tilde \rho_\psi \right) = - \widetilde{Tr} \left(\tilde \rho_\psi \ln \tilde \rho_\psi \right)$.  On the other hand, the connection of $tr$ to our path integral means that, when $|\psi\rangle$ is defined by a finite linear combination of surfaces to which the Lewkowycz-Maldacena argument \cite{Lewkowycz:2013nqa} can be applied in an appropriate semiclassical limit, our entropy $S(B)$ in that limit will conicide with that computed by the Ryu-Takayanagi formula \cite{Ryu:2006bv,Ryu:2006ef}.  In this sense we obtain a Hilbert space interpretation of the Ryu-Takayangi formula for general states having the $(d-2)$-dimensional closed source-manifold $B$ as any part of their boundary.

\section{Discussion}
\label{sec:disc}

The main primary goal of our work above was to generalize the analysis in \cite{Colafranceschi:2023urj} of UV-completions of asmptotically locally AdS Euclidean gravitational path integrals to Hilbert space sectors 
 $\mathcal H_{B_L\sqcup B_R}$ defined by bipartite boundaries ${B_L\sqcup B_R}$ for general $B_L, B_R$.  In analogy with the diagonal case $B_L=B_R$, the defining path integral provides  a notion of entropy on both $B_L$ and $B_R$.  Furthermore, by the Lewkowycz-Maldacena argument \cite{Lewkowycz:2013nqa}, if the path integral admits an appropriate semi-classical limit then the entropy in that limit is given by the Ryu-Takayanagi formula with small corrections.  Our main result is that
  there are Hilbert spaces $\undertilde{\mathcal H}_{B_L}, \undertilde{\mathcal H}_{B_L}$ such that 
$\mathcal H_{B_L\sqcup B_R}$ can be embedded in $\undertilde{\mathcal H}_{B_L} \otimes  \undertilde{\mathcal H}_{B_R}$  so that the above entropy on $B_L$ can be computed by tracing out 
$\undertilde{\mathcal H}_{B_R}$ to define a density matrix $\tilde \rho$ and then calculating $-\widetilde{Tr} \left(\tilde \rho  \ \ln \tilde \rho\right)$ using the trace $\widetilde{Tr}$ defined by summing  diagonal matrix elements of operators over an orthonormal basis of $\undertilde{\mathcal H}_{B_L}$.  In this sense we have provided a Hilbert space interpretation of Ryu-Takayanagi entropy without assuming the existence of a holographic dual CFT.

At the technical level, we showed that the non-diagonal sectors admit left and right von Neumann algebras, each of which can be obtained by appying central projections to the corresponding diagonal von Neumann algebras on $\mathcal H_{B_L\sqcup B_L}$ and $\mathcal H_{B_R\sqcup B_R}$.  This property ultimately followed from our use of rimmed source-manifolds, which can always be well-approximated by the product of a cylinder ($B_L \times [0,1]$ and/or $B_R \times [0,1]$) with another manifold as illustrated for the source-manifold $\psi$ in figure \ref{fig:phi}. An immediate implication of the above result is that both the left and right von Neumann algebras are type-I.
We also showed the left and right von Neumann algebras on any non-diagonal sector to be commutants on
$\mathcal H_{B_L\sqcup B_R}$.  As in \cite{Colafranceschi:2023urj}, we did not need to assume any of the various Hilbert spaces to be separable, though in realistic models one might expect that to be the case.

Natural directions for further research include providing a corresponding Lorentz-signature analysis and/or studying the effect of dropping the reality condition from the list of axioms.  Such reality conditions imply that the theory is invariant under time-reversal which, based on analogy with quantum field theory, one does not expect to hold in general.
Since the arguments both here and in \cite{Colafranceschi:2023urj} generally involved only the operation $\star$ (as opposed to separately using either the transpose operation ${}^t$ or the complex-conjugation operator ${}^*$), we suspect that it will be straightforward to drop this axiom from the list of requirements.  However, this remains to be checked in detail.

It is perhaps more important to emphasize that our work focused on bipartitions of the boundary, meaning that the boundary was written as a disjoint union of precisely {\it two} pieces $B_L\sqcup B_R.$  It thus remains to study more complicated partitions of the boundary.  In particular, for boundaries of the form $B_1 \sqcup B_2 \sqcup B_3$, we would not necessarily expect $\mathcal H_{B_1 \sqcup B_2 \sqcup B_3}$ to embed in a useful way into the tensor product 
$\undertilde{\mathcal H}_{B_1} \otimes  \undertilde{\mathcal H}_{B_2}\otimes  \undertilde{\mathcal H}_{B_3}$ of our Hilbert spaces 
$\undertilde{\mathcal H}_{B_1}$, $\undertilde{\mathcal H}_{B_2}$, $\undertilde{\mathcal H}_{B_3}$.  It would be very interesting to understand if our definitions of these spaces can be modified in a way that makes the above property hold.  Establishing that this is the case would bring us one step closer to showing that AdS/CFT-like results hold in a generic theory of gravity satisfying the axioms of \cite{Colafranceschi:2023urj}.

\acknowledgments

DM thanks Eugenia Colafranceshi, Xi Dong, and Zhencheng Wang for many related discussions. 
 He also acknowledges support from NSF grant PHY-2107939 and by funds from the University of California. DZ thanks Zhencheng Wang for the insightful explanations regarding their preceding work and the assistance provided in enhancing his understanding of the mathematical context. He also thanks Shanwen Jiang, Morgan Qi, Drew Rosenblum and Bartek Czech for fruitful conversations. This research was supported in part by grant NSF PHY-2309135 to the Kavli Institute for Theoretical Physics (KITP).

\appendix

\section{The off-diagonal sector is a Hilbert semi-birigged space that satisfies the coupling condition}
\label{app:Riefel}

This appendix provides the details associated with the use of theorem 1.3 from \cite{Rieffel:1976} for Hilbert semi-birigged spaces as required for the arguments of section \ref{sec:proj}.  In particular, it was stated in section \ref{sec:proj} that, when equipped with representations $C:=\hat A^{LR}_L$ and $D:=\hat A^{LR}_R$, the off-diagonal dense subspace $X:=\mathcal D_{LR}$ satisfies the four axioms of Hilbert-semi-birigged spaces found in \cite{Rieffel:1976}, and that it also satisfies the so-called coupling condition from that reference. We explain these properties briefly below.

The axioms introduced in \cite{Rieffel:1976} for a Hilbert semi-birigged space require the existence of two sesquilinear forms on $X$ that we call $(,)_C$ and $(,)_D$, and which take values in $C$ and $D$ as indicated.  In our case, we take these to be given by the left and right gluing of any of their representative surfaces (or linear superpositions thereof) with one surface involuted according to
\begin{eqnarray}
\label{eq:ses}
    \nonumber \left(\ket x,\ket y\right)_C&=&[xy^\star]_L=\hat x_L\hat y_L^\dagger=\Psi_L\left(\ket x\right)\Psi_L\left(\ket y\right)^\dagger \ \ \ {\rm and}\\
    \left(\ket x,\ket y\right)_D&=&[x^\star y]_R=\hat y_R\hat x_R^\dagger=\Psi_R\left(\ket y\right)\Psi_R\left(\ket x\right)^\dagger.
\end{eqnarray}

There are then 6 axioms to check:

\begin{enumerate}

\item $C$ and $D$ are faithfully represented on $X$.

\item $( x,y )_C z = x(y,z)_D$ for all $x,y,z\in X$.

\item If $c\in C$ is of the form $\left( x, y\right)_C$, then so is $c^\dagger$.

\item The linear span of the objects of the form $\left( x, y\right)_C z$ is dense in $X$.

\item $\left( x, y \right)_C = \left( y, x \right)_C$ and
$\left( x, y \right)_D  = \left( y, x \right)_D$ for all $x,y\in X.$

\item For any $x\in X$, both $( x,x)_C$ and  $( x,x)_D$ act as non-negative operators on $X$.
\end{enumerate}

Five of these axioms are essentially trivial to verify.
Axiom 1 was established in section \ref{sec:alg}.  Axioms 2 and 5 follow from short computations using our definitions \eqref{eq:ses}.  Axioms 3 and 6 are also manifest from
\eqref{eq:ses}.

This leaves only axiom 4.  To verify this remaining axiom, we need only show\footnote{This argument was inspired by the related proof of proposition I.9.2 in \cite{Takesaki}.} that there is no state $|\xi \rangle \in \mathcal H_{LR}$ that is orthogonal to all states of the form
$( x, y)_C z$ for $x,y,z\in \mathcal D_{LR}$.  To do so, we recall that
$\mathcal D_{LR}$ is dense in $\mathcal H_{LR}$, so such a $|\xi\rangle$ would require there to be a sequence $\{|\xi_n\rangle \in \mathcal D_{LR}\}$ that converges in norm to $\xi.$  Failure of axiom 4 would in particular require $|\xi\rangle$ to be orthogonal to all states of the form $\left( |\xi_n\rangle, |\xi_m \rangle \right)_C|\xi_k \rangle$, so that
\begin{equation}
\label{eq:xineq}
0 = \langle \xi| \left( |\xi_n\rangle, |\xi_m\rangle \right)_C |\xi_k \rangle =  \langle \xi| \Psi_L( |\xi_n\rangle) \left[ \Psi_L (|\xi_m\rangle)\right]^\dagger |\xi_k \rangle \ \ \ {\rm for \ all} \ n,m,k,
\end{equation}
where in the last step we have used \eqref{eq:ses} and the analogue of~\eqref{eq:adjpsiR} for the left map $\Psi_L$.
But we can use the weak-operator-topology continuity of the map $\Psi_L$ (the left version of Claim \ref{thm:strong})  and of the adjoint map $\dagger$ to take the limits $n,m,k\to \infty$ (in any order) of the matrix elements \eqref{eq:xineq} and obtain
\begin{eqnarray}
\label{eq:xieq}
0&=& \langle \xi| \Psi_L(|\xi \rangle) \left[\Psi_L(|\xi \rangle)\right]^\dagger  |\xi\rangle =
\lim_{\beta \downarrow 0}
\langle \xi|  \Psi_L(|\xi \rangle) \widehat{C_{2\beta}}_R \left[\Psi_L(|\xi \rangle)\right]^\dagger  |\xi\rangle
\cr
&=&\lim_{\beta \downarrow 0} \langle \xi|\Phi^R_{LR}(\widehat{C_\beta}_R)  \Psi_L(|\xi \rangle)  \left[\Psi_L(|\xi \rangle)\right]^\dagger \Phi^R_{LR}(\widehat{C_\beta}_R) |\xi\rangle
\cr
&=&\lim_{\beta \downarrow 0} \langle {C_\beta}| \left[\Psi_L(|\xi \rangle)\right]^\dagger \Psi_L(|\xi \rangle)  \left[\Psi_L(|\xi \rangle)\right]^\dagger \Psi_L(|\xi \rangle) |C_\beta\rangle
\cr
&=&tr(\{\left[\Psi_L(|\xi \rangle)\right]^\dagger \left[\Psi_L(|\xi \rangle) \right]\}^2),
\end{eqnarray}
where the 2nd equality follows by inserting the identity operator in the form $\lim_{\beta \downarrow 0} \widehat{C_{2\beta}}_R$, the 3rd uses the intertwining relation \eqref{eq:intCe}, the 4th uses~\eqref{eq:phi} twice, and the final step then follows from the definition \eqref{eq:trvN} of the trace on positive elements of the von Neumann algebra $\mathcal A_L^{LL}$.

Now, as in the proof of Claim \eqref{thm:nonzero}, since $0 \neq \langle \xi|\xi\rangle = tr\left(\left[\Psi_L(|\xi \rangle)\right]^\dagger] \Psi_L(|\xi \rangle) \right)$, the operator
$\left[\Psi_L(|\xi \rangle)\right]^\dagger] \Psi_L(|\xi \rangle)$ cannot vanish.  And since this operator is self-adjoint, its square must also be non-zero.  Faithfulness of the trace (as derived in \cite{Colafranceschi:2023urj}) then requires that \eqref{eq:xieq} be non-zero, contradicting \eqref{eq:xieq}. We thus see that  states of the form
$( x, y)_C z$ are dense in $\mathcal D_{LR}$, and thus that full set of axioms is satisfied.

In order to use theorem 1.3 from \cite{Rieffel:1976}, we will
also need to verify the so-called coupling condition. This condition states that if \( m, n \in X \) and \( x, y \in X \), and if 
\begin{equation}
\label{eq:coup}
    \langle m (x,z)_D, w \rangle = \langle z, n (y, w)_D \rangle \text{ for all } z, w \in X,
\end{equation}
then for any fixed \( z, w \in X \) there is a net \( \{ c_\alpha \} \) of elements of \( C \) such that
\begin{eqnarray}
\label{eq:coupnet}
c_\alpha^\dagger z & \rightarrow & m ( x, z)_D, \cr
c_\alpha w & \rightarrow & n ( y, w )_D.
\end{eqnarray}

To see that the coupling condition holds in our context, let us use equation~\eqref{eq:trprod} to rewrite the two sides of equation~\eqref{eq:coup} in the form:
\begin{eqnarray}
    \label{eq:coup1_5}
    tr\left((m (x,z)_D)^\star w\right)&=&tr\left( (x,z)_D^\star m^\star w\right)
    =tr((x^\star z)^\star m^\star w) \cr &=& tr(z^\star x m^\star w)=\braket{z|\hat x_L\hat m_L^\dagger|w} \ \ \ {\rm and} \cr
    tr(z^\star n \braket{y, w}_D)&=&tr(z^\star n y^\star w)=\braket{z|\hat n_L \hat y_L^\dagger|w}.
\end{eqnarray}
Thus we have
\begin{equation}
\label{eq:coup2}
\braket{z|\hat x_L\hat m_L^\dagger|w} = \braket{z|\hat n_L \hat y_L^\dagger|w}.
\end{equation}

Since equation~\eqref{eq:coup2} holds for all $z,w\in \mathcal D_{LR}$, we must have $\hat x_L\hat m_L^\dagger = \hat n_L \hat y_L^\dagger$. In particular, both sides are elements of $C=\hat A_L^{LR}$.   We are therefore free to take $\{c_\alpha\}$ to be the constant ($\alpha$-independent) net $c_\alpha=\hat x_L\hat m_L^\dagger = \hat n_L \hat y_L^\dagger$, for which we trivially compute the limits
\begin{eqnarray}
\label{eq:coupnet2}
c_\alpha^\dagger z & \rightarrow & \hat m_L \hat x^\dagger_L z,  \cr
c_\alpha w & \rightarrow & \hat n_L \hat y^\dagger_L w.
\end{eqnarray}
It is then straightforward to see the result \eqref{eq:coupnet2} agrees with the condition \eqref{eq:coupnet} by writing
\begin{eqnarray}
    m (x, z)_D=\hat z_R \hat x_R^\dagger m=|mx^\star z\rangle=\hat m_L \hat x^\dagger_L z\cr
    n ( y, w )_D=\hat w_R \hat y_R^\dagger n=|ny^\star w\rangle=\hat n_L \hat y^\dagger_L w,
\end{eqnarray}
where the 3rd expression in each line has been written by choosing representatives of each object in $\underline{Y}^d_{B_L \sqcup B_R}$.

\section{A generalization of $\Psi_R$}
\label{app:Psi}

As described in section \ref{sec:univ}, the construction of our map $\Psi_R: \mathcal H_{LR} \rightarrow \mathcal B(\mathcal H_{LL}, \mathcal H_{LR})$ generalizes readily to define a map  $\Psi_R^{B, B_1 \rightarrow B_2}: \mathcal H_{B_1 \sqcup B_2} \rightarrow \mathcal B(\mathcal H_{B \sqcup B_1}, \mathcal H_{B \sqcup B_2})$ involving arbitrary Hilbert spaces $B, B_1, B_2$.  The original map $\Psi_R$ is then the special case of $\Psi_R^{B, B_1 \rightarrow B_2}$ for which $B = B_1 = B_L, B_2 = B_R$.  The more general map $\Psi_R^{B, B_1 \rightarrow B_2}$ naturally satisfies properties analogous to those of the original $\Psi_R$.  The purpose of this appendix is to state those properties explicitly and to describe the corresponding proofs.  Our treatment below will be brief since we have already provided detailed arguments for $\Psi_R$ in the main text.

We formalize the main result of this appendix as follows:
\begin{theorem}
\label{thm:Psi3}
There is a map $\Psi_R^{B, B_1 \rightarrow B_2}:  \mathcal H_{B_1 \sqcup B_2} \rightarrow \mathcal B(\mathcal H_{B \sqcup B_1}, \mathcal H_{B \sqcup B_2})$ that is continuous w.r.t. the norm topology on $\mathcal H_{B_1 \sqcup B_2}$ and the strong operator topology on $\mathcal B(\mathcal H_{B \sqcup B_1}, \mathcal H_{B \sqcup B_2})$ which, if we define $\hat a_R : = \Psi_R^{B, B_1 \rightarrow B_2}(|a\rangle)$, for $|a\rangle \in \mathcal D_{B_1 \sqcup B_2}$ and $|b\rangle \in \mathcal D_{B \sqcup B_1}$ satisfies
\begin{equation}
\label{eq:Psi3}
\hat a_R |b \rangle = |ba\rangle
\end{equation}
for any representatives $a \in H_{B_1 \sqcup B_2},b\in H_{B \sqcup B_1}$ of the equivalence classes defined by $|a\rangle, |b\rangle$.  The map $\Psi_R^{B, B_1 \rightarrow B_2}$ is then uniquely defined by \eqref{eq:Psi3} and the above continuity requirements. 
 Furthermore, it satisfies the intertwining relation
\begin{equation}
\label{eq:Bint}
\Phi^L_{B B_2}(d) \hat a_R =  \hat a_R  \Phi^L_{B B_1}(d) 
\end{equation}
for all $d \in \mathcal A_L^{B \sqcup B}$. Here $\Phi^L_{BB_2}, \Phi^L_{BB_1}$ are just $\Phi^L_{LR}$ for $B_L=B$ and $B_R= B_2, B_1$.

There is also a corresponding left-map $\Psi_L^{B_2 \rightarrow B_1, B}: \mathcal H_{B_1 \sqcup B_2} \rightarrow \mathcal B(\mathcal H_{B_2 \sqcup B}, \mathcal H_{B_1 \sqcup B})$.
\end{theorem}

\begin{proof}
The proof of this claim is directly analogous to our arguments for the corresponding properties of the original $\Psi_R: \mathcal H_{LR} \rightarrow \mathcal B(\mathcal H_{LL}, \mathcal H_{LR})$.  The key point is that, because $\Psi_R$ defines an operator in $\mathcal B(\mathcal H_{LL}, \mathcal H_{LR})$ that acts at the \textit{right} boundary of $\mathcal H_{LL}$, the left boundary of this $\mathcal H_{LL}$ plays no role in the arugments and may be replaced with an arbitrary boundary $B$.  In particular, using footnote \ref{foot:genTrIn} (with $a$ replaced by $a^\star$) and the corresponding version of the trace inequality, for all $a \in H_{B_1 \sqcup B_2}, b \in H_{B \sqcup B_1}$ we have
\begin{equation}
\langle ba | ba\rangle \le \langle a | a\rangle \langle b | b \rangle, 
\end{equation}
so that \eqref{eq:Psi3} requires that $\hat a_R=0$ when $|a\rangle$ is a null state and, in addition, for all states $|a\rangle$ \eqref{eq:Psi3} implies that $\hat a_R$ annihilates all null states in $\mathcal H_{B \sqcup B_1}$.  The condition \eqref{eq:Psi3} is thus well-defined for $|a\rangle \in \mathcal D_{B_1 \sqcup B_2}$ and $|b\rangle \in \mathcal D_{B_1 \sqcup B_2}$.  Furthermore, the operator norm of $\Psi_R^{B, B_1 \rightarrow B_2}(|a\rangle)$ is bounded by $\langle a|a\rangle$ and so admits a unique continuous extension to all $|b\rangle \in \mathcal H_{B \sqcup B_1}$.  

Continuity in the argument $|a\rangle$ of $\Psi_R^{B, B_1\rightarrow B_2}$ then follows as in the proof of Claim \ref{thm:strong}, implying uniqueness of the corresponding extension to all $|a\rangle \in \mathcal H_{B_1 \sqcup B_2}$.  The intertwining relation \eqref{eq:Bint} can then be derived  by noting that for $a \in H_{B_1 \sqcup B_2}, d \in H_{B \sqcup B}$, and $\phi \in H_{B \sqcup B_1}$ we have $\hat d\in \hat A_L^{B\sqcup B}$ and
\begin{equation}
\label{eq:BintD}
\Phi^L_{B B_2}(\hat d) 
\hat a_R |\phi\rangle = |d\psi a \rangle 
=  \hat a_R \Phi^L_{B B_1}(\hat d) |\phi \rangle. 
\end{equation}
Since $\Phi_{B B_2}(\hat d)$, $\Phi_{B B_1}(\hat d)$, and $\hat a_R$  are bounded operators, they are continuous.  The left and right sides of \eqref{eq:BintD} must thus in fact  agree for all $|\phi \rangle \in \mathcal H_{B \sqcup B_1}$.    Continuity of the maps  
$\Phi_{B B_1}, \Phi_{B B_2}, \Psi_R^{B, B_1 \rightarrow B_2}$ with respect their arguments then similarly requires the left and right sides of \eqref{eq:BintD} to agree for all  $a \in \mathcal H_{B_1 \sqcup B_2}, d \in \mathcal H_{B \sqcup B}$.  This yields the desired intertwining relation \eqref{eq:Bint}. Note that
this simple argument was not available when Claim \ref{thm:intw} was originally proven in section \ref{sec:psi} since
continuity of $\Phi_{LR}$ had not yet been established.
\end{proof}

We will also use this appendix to state a small additional observation.  Before doing so, let us first recall the operation $\star$ defined on surfaces in any $Y^d_{B_1\sqcup B_2}$, where for $a \in Y^d_{B_1\sqcup B_2}$ we have $a^\star \in Y^d_{B_2\sqcup B_1}$.  Note that $|a\rangle$ and $|a^\star\rangle$ have the same norm.  Thus $\star$ is an anti-linear continuous map that extends to all states $|a\rangle \in \mathcal H_{B_1\sqcup B_2}$.  The output of this map will be denoted $|a^\star \rangle \in \mathcal H_{B_2\sqcup B_1}$.  This leads to the following result.

\begin{theorem}
\label{thm:fullPsiadj}
    The maps $\Psi_R^{B, B_1 \rightarrow B_2}:  \mathcal H_{B_1 \sqcup B_2} \rightarrow \mathcal B(\mathcal H_{B \sqcup B_1}, \mathcal H_{B \sqcup B_2})$ and $\Psi_R^{B, B_2 \rightarrow B_1}:  \mathcal H_{B_2 \sqcup B_1} \rightarrow \mathcal B(\mathcal H_{B \sqcup B_2}, \mathcal H_{B \sqcup B_1})$ intertwine the above anti-linear map $\star:  \mathcal H_{B_1 \sqcup B_2} \rightarrow \mathcal H_{B_2 \sqcup B_1}$ with the adjoint operation on $B(\mathcal H_{B \sqcup B_1}, \mathcal H_{B \sqcup B_2})$.  In other words, for $|a\rangle \in \mathcal H_{B_1 \sqcup B_2}$ and defining $\hat a_R: = \Psi_R^{B, B_1 \rightarrow B_2}\left( |a\rangle \right)$, $\widehat {\left(a^\star\right)}_R: = \Psi_R^{B, B_2 \rightarrow B_1}\left( |a^\star\rangle \right)$, we have
\begin{equation}
\label{eq:fullPsiadj}
    \widehat {\left(a^\star\right)}_R: =\hat a_R^\dagger.
\end{equation}
    \end{theorem}
\begin{proof}
    For $a\in \mathcal D_{B_1 B_2}$ this is just \eqref{eq:adjpsiR}.  Continuity of $\star$, $\Psi_R^{B, B_1 \rightarrow B_2}$, and $\Psi_R^{B, B_2 \rightarrow B_1}$ then imply the result for all $|a\rangle.$
\end{proof}

\section{Relating the left and right diagonal von Neumann algebras}
\label{app:Ldef}

Recall that our von Neumann algebras $\mathcal A_R^{B\sqcup B}$, 
$\mathcal A_L^{B\sqcup B}$ were defined by completing simpler algebras $\hat A_R^{B\sqcup B}$, 
$\hat A_L^{B\sqcup B}$ that were defined by concatenation of surfaces. The identification with surfaces then gives a natural map $\hat {\mathcal L}$ from $\hat A_R^{B\sqcup B}$ to 
$\hat A_L^{B\sqcup B}$.  In particular, any $\hat a_R \in \hat A_R^{B\sqcup B}$ is defined to act via the right-gluing of some $a \in \underline Y^d_{B \sqcup B}$ in the sense that, for $|\phi \rangle \in \mathcal D_{B \sqcup B}$, we have
\begin{equation}
 \hat a_R |\phi \rangle = |\phi a \rangle.   
\end{equation}
We may then define $\hat {\mathcal L}(\hat a_R) : = \hat a_L$, where  
\begin{equation}
 \hat a_L |\kappa \rangle = |a \kappa \rangle  
\end{equation}
for $|\kappa \rangle \in \mathcal D_{B \sqcup B}$.
For $|\phi\rangle, |\kappa\rangle \in
 \mathcal D_{B \sqcup B}$,  we find the relation 
\begin{equation}
\label{eq:Lprop0}
\langle \phi^\star| \hat a_L |\kappa^\star \rangle = tr(\phi a \kappa^\star) = tr(\kappa^\star \phi a )=  \langle \kappa| \hat a_R |\phi \rangle,  
\end{equation}
where we recall that the trace $tr$ is defined on the entire space $\underline Y^d_{B \sqcup B}$ and not just on the manifestly positive such elements.  We also find the relation
\begin{equation}
\label{eq:Lprop1}
\hat \kappa_R \hat a_R |\phi \rangle = |\phi a \kappa\rangle = \hat \phi_L \hat a_L |\kappa \rangle.
\end{equation}

We will use the relation \eqref{eq:Lprop0} to extend $\hat {\mathcal L}$ to a map $\mathcal L$ defined on all of $\mathcal A_R^{B \sqcup B}$.  In particular, for $a_R \in \mathcal A_R^{B \sqcup B}$ we define $a_L : = \mathcal L(a_R)$ to be the operator in  $\mathcal A_R^{B \sqcup B}$ whose matrix elements satisfy
\begin{equation}
\label{eq:Lprop2}
\langle \phi| a_L |\kappa \rangle = \langle \kappa^\star | a_R |\phi^\star \rangle  
\end{equation}
for all $|\phi\rangle, |\kappa \rangle \in \mathcal H_{B \sqcup B}$.  

To see that $a_L$ is a bounded operator on $\mathcal H_{B \sqcup B}$, let us introduce an orthonormal basis $|i \rangle$ and write
\begin{eqnarray}
\label{eq:Lprop3}
|  a_L |\kappa \rangle|^2  &=& 
\sum_i |\langle i | a_L |\kappa \rangle|^2 \cr
&=& \sum_i |\kappa^\star |  a_R |i^\star \rangle|^2 \cr &=& |  a_R |\kappa^\star \rangle|^2 \le ||a_R||^2 \langle \kappa | \kappa \rangle.
\end{eqnarray}
Furthermore, to see that $a_L \in \mathcal A_L^{B \sqcup B}$, let us write $a_R$ as the weak-operator-topology limit of a net $\{\widehat {(a_\alpha)}_R \} \subset \hat A_R^{B \sqcup B}$.  We then simply note that for $\widehat {(a_\alpha)}_L : = \mathcal L \left( \widehat {(a_\alpha)}_R \right)$, the net $\{\widehat {(a_\alpha)}_L \} \subset \hat A_L^{B \sqcup B}$ converges in the weak operator topology to $a_L.$

We now wish to show the following:
\begin{theorem}
    For any boundary $B$, any $a_R \in \mathcal A_R^{B \sqcup B}$, and any states $|\phi\rangle, |\kappa\rangle \in \mathcal H_{B \sqcup B}$ the operator $a_L : = \mathcal L(a_R)$ satisfies 
\begin{equation}
\label{eq:Lprop}    
\hat \kappa_R a_R |\phi \rangle = |\phi a \kappa\rangle = \hat \phi_L a_L |\kappa \rangle.  
\end{equation}
\end{theorem}
\begin{proof}
To see this, we again consider a net $\{\widehat {(a_\alpha)}_R \} \subset \hat A_R^{B \sqcup B}$ that converges to $a_R$ in the weak operator topology and the associated net $\{\widehat {(a_\alpha)}_L \} \subset \hat A_L^{B \sqcup B}$ that converges to $a_L$. Using relation \eqref{eq:Lprop1} for each $\alpha$ and taking the inner product with a fixed state $|\gamma \rangle\in \mathcal H_{B \sqcup B}$ gives  
\begin{equation}
\label{eq:Lprop4}
\langle \gamma | \hat \kappa_R \widehat {(a_\alpha)}_R |\phi \rangle  = \langle \gamma |\hat \phi_L  \widehat {(a_\alpha)}_R |\kappa \rangle.
\end{equation}
Taking limits then gives 
\begin{equation}
\label{eq:Lprop5}    
\langle \gamma | \hat \kappa_R a_R |\phi \rangle = \langle \gamma |\phi a \kappa\rangle = \hat \phi_L a_L |\kappa \rangle  
\end{equation}
for all $|\gamma\rangle$, which is equivalent to \eqref{eq:Lprop}.
\end{proof}

\end{document}